\documentclass{article}

\usepackage{arxiv}

\usepackage[utf8]{inputenc} 
\usepackage[T1]{fontenc}    
\usepackage{url}            
\usepackage{booktabs}       
\usepackage{amsfonts}       
\usepackage{nicefrac}       
\usepackage{microtype}      
\usepackage{amsmath}
\usepackage{graphicx}
\usepackage[english]{babel} 
\usepackage{amssymb}
\usepackage{amsthm}
\usepackage{commath}
\usepackage{mathtools}
\usepackage{xcolor}
\usepackage{bm}
\usepackage{bbm}
\usepackage[ruled,vlined]{algorithm2e}
\usepackage{hyperref}
\usepackage[noabbrev]{cleveref}

\usepackage{dsfont}
\hypersetup{
	colorlinks=true,
	linkcolor=blue,
	citecolor=blue,
}
\usepackage{caption}
\usepackage{subcaption}

\usepackage[backend=bibtex, 
style=authoryear,
hyperref = true, 
backref=true,
doi=false, isbn=false, url=false, eprint=false, dashed=false,
maxcitenames=2]{biblatex} 

\newcommand{\pkg}[1]{{\fontseries{b}\selectfont #1}}
\newcommand{\bZ}{\mathbf{Z}}
\newcommand{\bX}{\mathbf{X}}
\newcommand{\R}{\mathbb{R}}
\newcommand{\Exponential}{\mathcal{E}}
\newcommand{\Poisson}{\mathcal{P}}
\newcommand{\Uniform}{\mathcal{U}}
\DeclareMathOperator{\Bernoulli}{\mathcal{B}}

\newtheorem{proposition}{Proposition}
\DeclareMathOperator{\ICLlin}{ICL_{\textit{lin}}} 
\DeclareMathOperator{\ICLex}{ICL_{\textit{ex}}} 
\DeclareMathOperator{\ICL}{ICL}
\DeclareMathOperator{\ICLbic}{ICL_{\textit{BIC}}}
\DeclareMathOperator{\Dir}{Dir}
\DeclareMathOperator{\Betad}{Beta}
\newcommand{\Mult}{\mathcal{M}}

\newcommand{\Partition}{\mathcal{P}}

\bmdefine{\bC}{C}
\newcommand{\p}{p}
\newcommand{\Simplex}{\Delta}

\newcommand{\nb}{n}
\newcommand{\K}{K}
\renewcommand{\dim}{p}
\bmdefine{\Obs}{\bm{X}}
\bmdefine{\obs}{\bm{x}}
\newcommand{\rawobs}{x}
\bmdefine{\Clust}{\bm{Z}}
\bmdefine{\clust}{\bm{z}}
\newcommand{\rawclust}{z}
\bmdefine{\scores}{\obs}

\bmdefine{\param}{\bm{\theta}}
\newcommand{\rawparam}{\theta}
\bmdefine{\globalparam}{\vartheta}

\bmdefine{\bPi}{\bm{\pi}}
\newcommand{\totalcount}{c}
\bmdefine{\bmean}{\bm{m}}
\bmdefine{\bS}{\bm{S}}
\bmdefine{\bD}{\bm{D}}
\bmdefine{\bDelta}{\Delta}

\bmdefine{\bBeta}{\beta}

\bmdefine{\balpha}{\bm{\alpha}}

\bmdefine{\Clustrow}{\bm{Z}_{r}}
\bmdefine{\clustrow}{\bm{z}}
\newcommand{\rawclustrow}{z}
\bmdefine{\Clustcol}{\bm{Z}_{c}}
\bmdefine{\clustcol}{\bm{z}}
\newcommand{\rawclustcol}{z}
\bmdefine{\bPirow}{\bm{\pi}_{r}}
\bmdefine{\bPicol}{\bm{\pi}_{c}}
\newcommand{\Kr}{{K_{r}}}
\newcommand{\Kc}{{K_{c}}}

\renewcommand{\Clustrow}{\bm{Z}^{r}}
\renewcommand{\clustrow}{\bm{z}^{r}}
\renewcommand{\rawclustrow}{z^{r}}
\renewcommand{\Clustcol}{\bm{Z}^{c}}
\renewcommand{\clustcol}{\bm{z}^{c}}
\renewcommand{\rawclustcol}{z^{c}}

\renewcommand{\bPirow}{\bm{\pi}^{r}}
\renewcommand{\bPicol}{\bm{\pi}^{c}}

\graphicspath{{./figures/}}

\newcommand{\revision}[1]{#1}

\title{Hierarchical clustering with discrete latent variable models and the integrated classification likelihood}

\author{Etienne C\^{o}me \\ COSYS/GRETTIA, Universit\'{e} Gustave-Eiffel, \\ Noisy-Le-Grand, France\\
			\And
			Nicolas Jouvin \\
			Universit\'{e} Paris 1 Panth\'{e}on-Sorbonne, SAMM, France\\
			FP2M, CNRS FR 2036, Paris, France \\
			\And 
			Pierre Latouche\\
			Universit\'{e} de Paris, MAP5, CNRS,  \\
			FP2M, CNRS FR 2036, Paris, France \\ 
			\And 
			Charles Bouveyron \\
			Universit\'{e} C\^{o}te d’Azur, Inria, CNRS, Laboratoire J.A. Dieudonné\\
			Maasai research team, Nice, France}

\begin{document}
\maketitle

	\begin{abstract}
Finding a set of nested partitions of a dataset is useful to uncover relevant structure at different scales, and is often dealt with a data-dependent methodology. In this paper, we introduce a general two-step methodology for model-based hierarchical clustering. Considering the integrated classification likelihood criterion as an objective function, this work applies to every discrete latent variable models (DLVMs) where this quantity is tractable. The first step of the methodology involves maximizing the criterion with respect to the partition. Addressing the known problem of sub-optimal local maxima found by greedy hill climbing heuristics, we introduce a new hybrid algorithm based on a genetic algorithm \revision{which allows to efficiently explore the space of solutions}. The resulting algorithm carefully combines and merges different solutions, and allows the joint inference of the number $K$ of clusters as well as the clusters themselves. Starting from this natural partition, the second step of the methodology is based on a bottom-up greedy procedure to extract a hierarchy of clusters. In a Bayesian context, this is achieved by considering the Dirichlet cluster proportion prior parameter $\alpha$ as a regularization term controlling the granularity of the clustering. A new approximation of the criterion is derived as a log-linear function of $\alpha$, enabling a simple functional form of the merge decision criterion. This second step allows the exploration of the clustering at coarser scales. The proposed approach is compared with existing strategies on simulated as well as real settings, and its results are shown to be particularly relevant. A reference implementation of this work is available in the \pkg{R} package 
\pkg{greed} accompanying the paper\footnote{available at \href{http://github.com/comeetie/greed} {http://github.com/comeetie/greed}}. 
	\end{abstract}

	\section{Introduction}
	\label{HCICL:intro}
	Partitional approaches to clustering seek a partition $\Partition$ of the observations into $\K$ class, or clusters. Hierarchical clustering extends this idea by seeking a path of nested partitions, from finer to coarser. Usually dealt with \textit{ad hoc} similarity functions depending on the data at hand, this approach is popular in unsupervised data analysis such as biological taxonomy, phylogenetics or social network analysis for uncovering hierarchical structures \parencite[Section 4]{everitt2011cluster}. In this paper, we address the problem of building hierarchies of partitions \revision{in a unified approach within the statistical framework of model-based clustering}. 

	\subsection{Model-based clustering with discrete latent variable models}
	
	Model-based clustering is a principled approach for clustering, with a variety of flexible models depending on the data at hand \parencite{bouveyron2019model}. Aiming at understanding the different sources of randomness in observations, it can therefore help in interpreting the clusters uncovered in practice. In this paper, we consider a general class of models used in model-based clustering that we call discrete latent variable models (DLVMs). This class encompasses finite mixture models \parencite{McLachlan2000}, which is the most popular instance of model-based approaches. Moreover, it also includes other popular models, which do not exactly fit the definition of finite mixtures. Popular examples are the stochastic block model (SBM) for network analysis \parencite{wang1987stochastic,nowicki2001estimation} and its extensions \parencite[see][for instance]{karrer2011stochastic}, as well as the latent block model \parencite[LBM,][]{govaert2010} for co-clustering. In model-based clustering, the partition is a latent variable $\Clust$, where each element of $\Clust$ is a binary vector of size $\K$ indicating clustering membership, and $\K$ denotes the number of clusters.
	The general definition of a DLVM assumes that the observations, denoted as $\Obs$, are drawn from a two-step process: first, the latent partition $\Clust$ is drawn independently from a product of multinomial distributions parameterized by $\bPi$. Then, the observations are supposed to be independent given the whole partition. The complete likelihood, known as the \textit{classification} likelihood in this context, is written as:
	\begin{align}
	\label{HCICL:eq:DLVM}
	\p(\Obs, \bZ\mid \bPi, \param) = 
	\prod_{\clust \in \Clust} \p(\clust \mid \bPi) \underset{\text{factorized}}{\underbrace{\prod_{\obs \in \Obs}\p(\obs \mid \Clust, \param)}},
	\end{align}
	where $(\bPi, \param)$ are a set of parameters respectively controlling the cluster membership and the conditional distributions.

	In the case of finite mixture models, the observations are $\nb$ independent random vectors $\Obs = \{\obs_1, \ldots, \obs_{\nb}\}$ in \revision{dimension $p$}, which can be summarized in a data matrix $\Obs $ \revision{of size $\nb \times \dim$}. In this context, each observation is assigned to a latent multinomial variable $\clust_i \in \{0,1\}^K$, defining its cluster assignment. The latter is independently drawn from a multinomial distribution, with proportions $\bPi=(\pi_1, \dots, \pi_\K)$. Then, an observation $\obs_{i}$ follows some conditional distribution depending on the value of $\clust_{i}$, and the sampling process for all $i$ is as follows:
	\begin{equation}
	\label{eq:MixtureModels}
	\begin{aligned}
	\clust_{i} \mid \bPi& \sim \mathcal{M}(1,\bPi),\\
	\obs_{i} \mid \rawclust_{ik}=1, \param_{k} &\sim \p(\obs_{i} \mid \param_{k}).
	\end{aligned}
	\end{equation}
	The parameters $\bPi$ controls the prior probability of belonging to each group, while the mixture parameters control the distribution in the $k$-th cluster, and depend on the observational model at hand. For instance, in a Gaussian mixture we have $ \p(\obs_{i} \mid \param_{k}) = \mathcal{N}(\obs_i \mid \bmean_k, \bS_k)$, respectively the mean and covariance matrix in cluster $k$.
	
	In the case of network analysis with the stochastic block model, the observations are the edges $\Obs = \{\rawobs_{ij}\}$, where $\rawobs_{ij}$ represents the presence of absence of an edge. It can be binary, $\rawobs_{ij} \in \{0, 1\}$, or weighted $\rawobs_{ij} \in \R$ \parencite{mariadassou2010uncovering}. Observing the edges, \textit{e.g.} the topology of the graph, we wish to cluster the nodes $\{1, \ldots, \nb\}$. Thus, each node $i$ is assigned to a cluster latent variable $\clust_{i}$ and the edges are supposed to be conditionally independent given the partition, with a conditional distribution depending only on the clusters of their out and end-nodes. The sampling process is then written as:
	\begin{equation}
	\label{eq:SBM}
	\begin{aligned}
	\rawobs_{ij} \mid \rawclust_{ik}\rawclust_{jl}=1, \param_{kl} &\sim \p(\rawobs_{ij} \mid \param_{kl}).
	\end{aligned}
	\end{equation} 
	The \textit{block} terminology stems from the assumption that edges with extremities in the same pair of clusters $(k,l)$ are independent and identically distributed, hence forming homogeneous block in the adjacency matrix $\Obs$. As in mixture models, the parameter $\bPi$ controls the group proportions and the latent partition is drawn independently from $\Mult_{\K}(1, \bPi)$. However, the SBM does not exactly fit the definition of finite mixture models, since the observations are no longer marginally independent. \textcite{matias2014modeling} discuss this fact using the moralized graphical model of SBM, revealing the posterior dependencies that do not arise in finite mixture models. The set of parameters $\param_{kl}$ is now specific to the pair $(k, l)$ of clusters, and depends on the specific model. For instance, in the case of a binary SBM, the block distributions are Bernoulli parameterized by $\rawparam_{kl} \in [0; 1]$. Degree-corrected versions have also been proposed to account for a strong degree heterogeneity inside blocks, displayed by real-world networks \parencite{karrer2011stochastic, zhu2014}. In these models the Bernoulli distribution of edges is replaced by a Poisson, the latter making a good approximation for the Bernoulli when the mean parameter is small \parencite{zhao2012consistency}.
	
	In co-clustering, the observations $\Obs = \{ \rawobs_{ij} \}$ are supposed to be given in a data matrix $\Obs \in \R^{\nb \times \dim}$, and one seeks a bipartition $\Clust = (\Clustrow, \Clustcol)$ with $\Kr$ clusters over the $\nb$ rows and $\Kc$ clusters over the $\dim$ columns. The latent block model \parencite[LBM,][]{govaert2010} supposes conditional independence of entries $\obs_{ij}$ given $\clustrow_{i}$ and $\clustcol_{j}$. The sampling scheme is given as:
	\begin{equation}
	\label{eq:LBM}
	\rawobs_{ij} \mid \rawclustrow_{ik} \rawclustcol_{jl}=1, \param_{kl} \sim \p(\rawobs_{ij} \mid \param_{kl}),
	\end{equation}
	and is very close to SBM. Indeed, the latter may be viewed as a particular instance of LBM when $\nb = \dim$ and $\Clustcol = \Clustrow$. Moreover, the row partition $\Clustrow$ and column partition $\Clustcol$ are supposed to be respectively drawn \textit{i.i.d.} from $\Mult_{\Kr}(1, \bPirow)$ and $\Mult_{\Kc}(1, \bPicol)$. Thus, the distribution of $\Clust$ is a product of multinomials parameterized by $\bPi = (\bPirow, \bPicol)$, hence fitting the definition of \Cref{HCICL:eq:DLVM}.
	
	
	Statistical inference constitutes the most popular approach to model-based clustering. In the one hand, the frequentist approach seeks to estimate the model parameters $(\bPi, \param)$ via maximum-likelihood estimation. As the partition $\Clust$ is considered to be a latent variable, the expectation-maximization (EM) algorithm \parencite{Dempster77} has become quite a universal tool, especially for finite mixture models \parencite{mclachlan2007algorithm}. For other DLVMs, variational extensions \parencite{blei2017variational} have been introduced to deal with the problem of the intractability of the posterior distributions. Notably, this is the case for the binary SBM \parencite{Daudin2008} and the LBM \parencite{govaert2010}.  On the other hand, the Bayesian paradigm considers model parameters as random variables with so-called prior distributions. In this context, statistical estimation seeks to sample from the posterior distributions of parameters, which is usually done via some Markov-Chain Monte-Carlo scheme \parencite[chap. 4]{fruhwirth2019handbook}. Clustering comes as a byproduct of these procedures, as one can then estimate $\Clust$ from its exact or approximated posterior distribution.

	\subsection{Exact ICL criterion and greedy hill climbing heuristics}
	\label{ssec:introICL}
	
	In addition to its flexibility, the statistical framework provides a sound way to choose between different values of $\K$ as a model selection problem. Specific statistical tools are required in this context, and several penalized likelihood criteria were proposed such as the Akaike information criterion \parencite[AIC,][]{akaike1974new}, or the Bayesian information criterion \parencite[BIC,][]{schwarz1978estimating}. For a recent and detailed review on model selection for model-based clustering, we refer to \textcite[chap. 8]{fruhwirth2019handbook}. In this paper, we focus on the ICL criterion, which was specifically introduced in the context of model-based clustering. \textcite{Biernacki2000} first used a combination of Laplace and Stirling approximations on $\log \p(\Obs, \Clust \mid \K)$ to find an asymptotic criterion called $\ICLbic$ due to its close link to the BIC. \revision{Since this criterion involves an estimation of $\Clust$, that is a hard assignment of observations to clusters, it is specifically designed for the clustering task. Thus, as shown  in \cite{Biernacki2010}, it is more robust to model miss-specification than BIC which focuses on the estimation of the model density. While $\ICLbic$ depends explicitly on an estimation of $\Clust$, it is still referred to as a model selection criterion in the literature and used in the model-based clustering context. Indeed, its derivation is dependent on the statistical models considered.}

	Recent works have also considered exact expressions of the ICL, putting a factorized conjugate prior distribution $\p(\bPi, \param \mid \balpha, \bBeta) = \p(\bPi\mid \balpha) \p(\param \mid \bBeta)$ over the model parameters, and defined as:
	\begin{equation}
	\label{eq:ICLex}
	\begin{aligned}
		\ICLex(\Clust; \balpha, \bBeta) &= \log\left(\int_{\param}\int_{\bPi}\p(\bX|\bZ,\param)\p(\param|\bBeta)\p(\bZ|\bPi)\p(\bPi|\balpha)\dif\param \dif\bPi\right), \\
		& =  \log \p (\Obs \mid \Clust, \bBeta) + \log \p (\Clust \mid \balpha).
	\end{aligned}
	\end{equation}
The $\balpha$ parameters control the conjugate distribution, which, in the case of a DLVM, is a Dirichlet\footnote{Or a product of Dirichlet distributions $\Dir_{\Kr}(\balpha_r) \times \Dir_{\Kc}(\balpha_c)$ in the case of co-clustering with the LBM. Except for an additional notation burden, the rest of the discussion easily extends to this case, which is discussed in detail in \cref{sec:models}} over group proportions $\Dir_{\K}(\balpha)$. This part is thus common to all DLVMs in the sense that it does not depend on the observational model on $\Obs$. If a symmetric Dirichlet is chosen, with $\alpha_k = \alpha$, the second term can be made explicit using the independence of the elements of $\Clust$:
\begin{equation}
\label{eq:ICLexdec}
\ICLex(\Clust; \alpha,\bBeta) = \log \p(\Obs \mid \Clust, \bBeta)  +  \log\left( \dfrac{\Gamma(\K \alpha)\prod\limits_{k=1}^{\K} \Gamma(\alpha + n_k)}{\Gamma(\alpha)^\K  \Gamma(\nb + \alpha\, \K)}\right) , 
\end{equation}
where $\Gamma(\cdot)$ is the Gamma function and $n_k=\sum_i \rawclust_{ik}$. Usually, the hyper-parameter $\alpha$ is set to $1$ or $\frac{1}{2}$ to specify uninformative a uniform or a Jeffreys prior. 

The hyper-parameters $\bBeta$ control the conjugate prior over the mixture parameters $\param$, and depends on the generative model at hand. Naturally, such criterion is restricted to particular DLVMs, where such conjugate distributions are easy to derive, so that the first term in \Cref{eq:ICLex} is analytic. However, this class is quite large and expressions are available for the mixture of multinomials \parencite{Biernacki2010} and Gaussian mixture \parencite{Bertoletti2015}, while being virtually feasible for any mixture of exponential families as they admit natural conjugate priors \parencite[p. 42]{gelman2004bayesian}. Exact ICL criteria were also derived for the SBM \parencite{Come2015}, the LBM \parencite{Wyse2017} and degree-corrected variants \parencite{Newman2016,Riolo2017}. For the sake of self-consistency and to ensure unified notations, \Cref{sec:models} contains expressions and technical derivations of $\ICLex$ for all the considered models, in particular for the directed formulation of SBM and its degree-corrected variants.

In between the frequentist and the Bayesian approaches, a new line of work started to consider direct maximization of $\ICLex$ with respect to the partition $\bZ$, avoiding the inference step over the parameters $(\bPi, \param)$ \revision{and tackling more directly the clustering task}. In order to solve this discrete and combinatorial optimization problem, greedy heuristics were successfully tested to directly optimize this criterion over the space of possible partitions. These approaches consist in hill climbing algorithms, starting from an initial partition and greedily swapping clusters until some local maximum is met. Eventually, at the end, some clusters are merged up to the point where no more merge moves can maximize the $\ICLex$. Such a type of algorithms performs model selection and clustering at the same time and are computationally attractive compared to approximate or exact inference alternatives. This approach dates back to \textcite{tessier2006evolutionary}, for the latent class model. It was then extended in \textcite{Come2015} for SBM, and applied to other DLVMs such as Gaussian mixture models \parencite{Bertoletti2015}, LBM \parencite{Wyse2017} and dynamic variants of SBM \parencite{Corneli2016,zreik2016dynamic}. \revision{The aforementioned greedy maximization procedure comes with a cost. In practice, the objective is highly multimodal, and the combinatorial nature of the search space multiplies the presence of local maxima in which the methods may be stuck. Different techniques were proposed to tackle this problem, including several restarts \parencite{Come2015}, batch versions \parencite{Bertoletti2015}, and genetic evolutionary algorithms \parencite[GAs,][]{tessier2006evolutionary}. This latter methodology borrows from biological evolution principles, combining solutions via crossover operators, allowing random modifications and discarding poor solutions, in an analogy to genetic inheritance, mutations and natural selection. In the context of $\ICLex$ maximization, they represent a promising avenue to efficiently explore the partition space, avoiding the pitfalls of sub-optimal local maxima through the recombination and mutation operators \parencite[p. 137]{fruhwirth2019handbook}.}

	\subsection{Contributions and organization of the paper}
	
	This paper builds on two main contributions to propose a two-step methodology for hierarchical clustering. 
	
	First, we address the issue of spurious local maxima in greedy maximization of $\ICLex$. We propose a hybrid genetic algorithm mixing an evolutionary strategy with local search to optimize the $\ICLex$ criterion, efficiently exploring the space of partitions. The novelty and efficiency of this approach resides in the representation of solutions as set partitions, and in the crossover operator used to recombine solutions, carefully preserving their structure. This algorithm, \revision{presented in \Cref{sec:HybridAlgo}}, is adaptable to a wide variety of DLVMs, as soon as swap and merge moves can be efficiently computed. \revision{Such an algorithm allows the inference of the number of clusters $K$ and of the clusters themselves. However, it usually leaves a partition with $K > 1$ and does not allow to build a hierarchy of partitions up to $K=1$.}
	
	Second, we introduce an agglomerative hierarchical algorithm considering $\ICLex$ as a function of the hyper-parameter $\alpha$ and relying on a new approximation, using the asymptotic of the log-Gamma function when $\alpha$ goes to $0$. We show that decreasing $\alpha$ can unlock fusions in the sense that coarser partitions achieve a greater ICL value. Starting from an ICL-dominant solution at a given level $\alpha$, typically $1$, the proposed heuristic extracts a set of nested clustering that are each dominants with respect to this new criterion over some range of $\alpha$ values. In addition, this strategy, \revision{presented in \Cref{sec:Hierarchical}}, enables the construction of a cluster dendrogram, giving a natural ordering for the clusters which is interesting for interpretation and visualization purposes, particularly on real datasets.
	
	These two contributions are generically applicable in the framework of DLVMs for which conjugate priors can be easily derived for the observational model. Moreover these algorithms are naturally linked, working with similar objectives, and the first one can be used as an initialization for the second to extract a hierarchical clustering. One of the particularity of this approach is that it only extracts the relevant part of the dendrogram, since the latter typically starts with an optimal partition obtained at $\alpha=1$ or $1/2$. Therefore, it avoids the analysis of uninformative fusions commonly encountered in the firsts stages of classical hierarchical agglomerative clustering algorithms. This approach is also computationally efficient and may handle large datasets which could be hard to grasp with classical fully hierarchical algorithms. 
	
	\Cref{sec:exp} gives a detailed investigation of the two algorithms behaviors on simulated and real datasets, along with a thorough comparison with related model-based clustering algorithms. Three types of DLVMs are considered: the mixture of multinomials, the SBM and LBM, as well as their degree-corrected versions.

	Finally, the open-source \pkg{R} package \parencite{Rcore} \pkg{greed} provides a reference implementation of the algorithms introduced in this paper. The implementation is extendable and new models can be integrated. The main computationally demanding methods were developed in \pkg{Cpp} thanks to the \pkg{Rcpp} package \parencite{Eddelbuettel2017} taking advantages of sparse matrix computational efficiency thanks to the \pkg{RcppArmadillo} and \pkg{Matrix} packages \parencite{Eddelbuettel2014,Bates2019}. Eventually, the \pkg{future} package \parencite{Bengtsson2019} was used to enable easy parallelization of the computations of the proposed hybrid genetic algorithm.
		
	\subsection{\revision{Motivating example}}
		
	As a motivating example for the proposed two-step methodology, we simulate a random SBM graph with $\nb = 1500$ nodes and a hierarchical cluster structure with 3 clusters each composed of 5 small clusters. The small clusters have an intra-connectivity probability of $0.1$ and a probability of connecting a node from the same cluster of $0.025$. Otherwise, two random nodes may be connected with a probability of $0.001$. Figure \ref{fig:motiv} compares the clustering results in the form of adjacency matrices obtained by different methods : a greedy hill climbing optimization with a random starting partition with twenty clusters, the proposed hybrid optimization algorithm and the same results after a reordering of the clusters with the hierarchical heuristic.	As clearly shown by this example, the greedy heuristic with a random starting point suffers from under-fitting with only four clusters extracted among the 15 simulated. The hybrid algorithm does not suffer from the same problem in this example, and recovers correctly the 15 simulated clusters. Finally, the hierarchical ordering enable a clear visualization of the hierarchical structure of this dataset, that is also clearly depicted in the extracted dendrogram presented in Figure~\ref{fig:motivdendo}.

			\begin{figure}[!ht]
		\centering
		\includegraphics[width=0.95\textwidth]{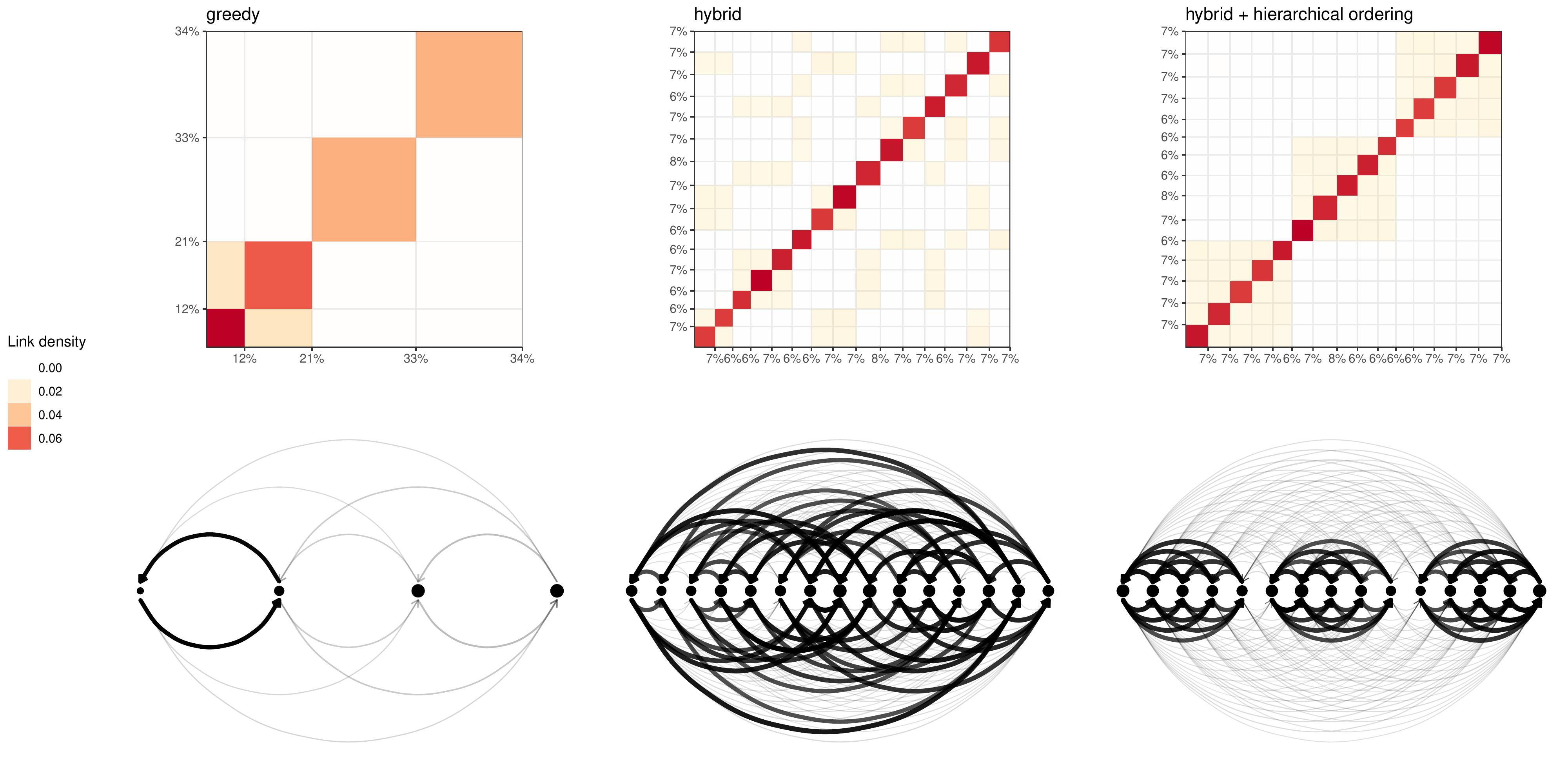}
		\caption{Motivating example for the proposed algorithms. Block matrix representation of the solutions (upper row) and cluster node link diagram (bottom row) obtained with (from left to right) a greedy algorithm with a random starting point, the proposed hybrid algorithm and the same clustering but with clusters re-arranged thanks to the hierarchical ordering.}
		\label{fig:motiv}
	\end{figure}
	\begin{figure}[!ht]
		\centering
		\includegraphics[width=0.45\textwidth]{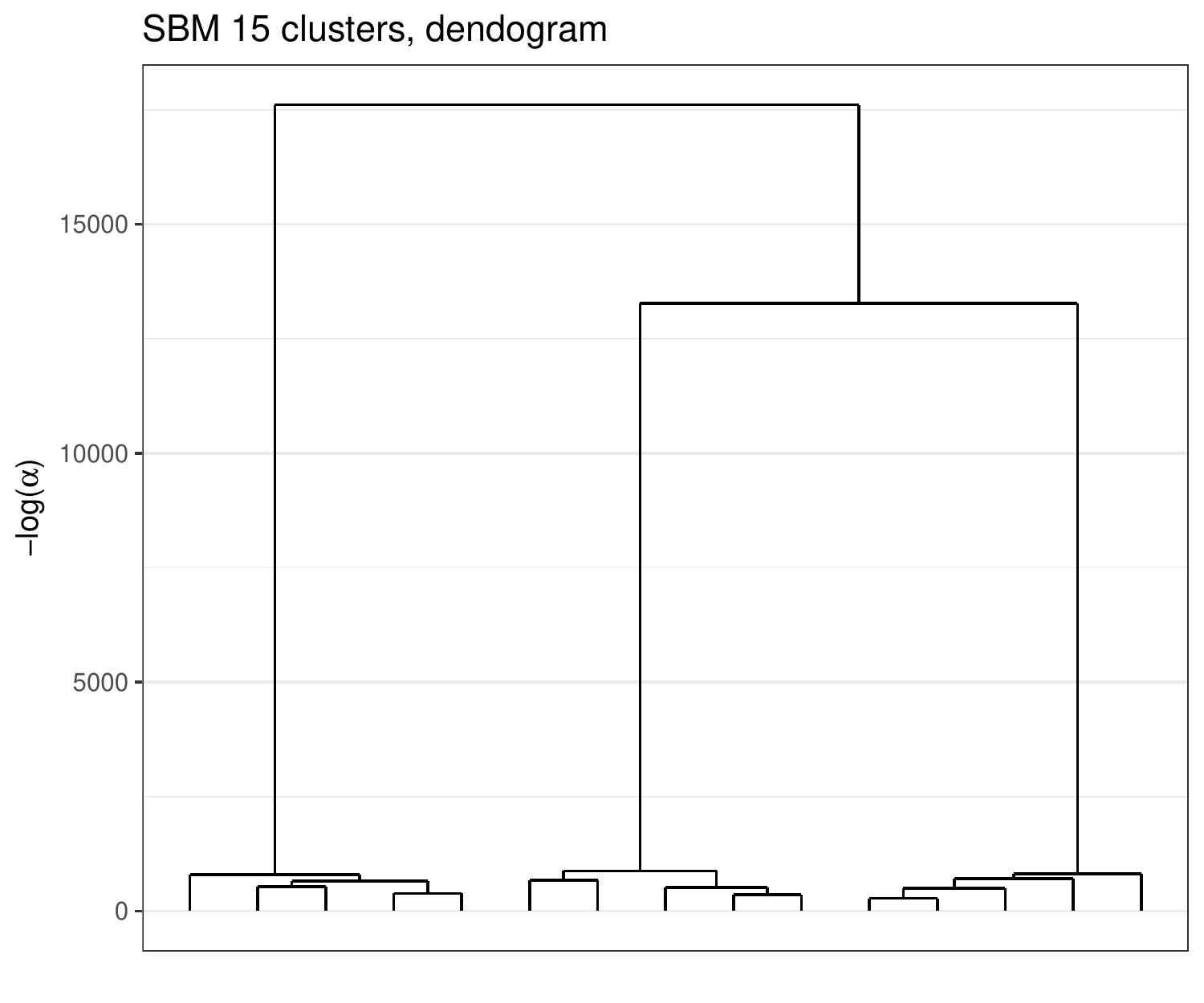}
		\caption{Motivating example for the proposed algorithms. Clusters dendrogram extracted with the hierarchical regularization path heuristic.}
		\label{fig:motivdendo}
	\end{figure}

	\section{A hybrid genetic algorithm for DLVMs}
	\label{sec:HybridAlgo}

	As explained above, several works rely on the $\ICLex$ criterion as an objective function to maximize with respect to the partition $\Clust$. These are mainly based on greedy hill climbing algorithms: starting from a carefully chosen over-segmented initial partition, or \textit{seed}, swaps and eventually merges are applied to increase the criterion. In addition to the competitive computational complexity and the ease of implementation, these algorithms may be seen as an automatic way to perform model selection, as clusters may be emptied during the process. In the SBM case, \textcite{Come2015} proposed a thorough comparison with state-of-the-art methods that illustrates the interest of such algorithms. 
	
	However, a major drawback of this approach is its dependency to the initialization. Indeed, defining a relevant initial partition is not trivial, and the method may lead to under-fitting as demonstrated in the introductory example in Figure~\ref{fig:motiv}. Here, the issue seems to lie in the lack of exploration of the partition space, and genetic algorithms (GAs) have been proposed to improve the exploration. Starting from a given solution, the latter evolves a population of candidate solutions by selecting some of the most promising ones, combining them, and mutating them until a specified number of generations or some stopping criterion is met. As described in \textcite[Chapter 2]{Eiben2003}, the fundamental components of such algorithms are the solution representation, the selection strategy and the variation operators used for recombination and mutation. However, while GAs are very good at identifying near-optimal regions of the search space, they can take a relatively long time to reach a local optimum in the region of interest. In order to improve their exploitation capacity, a number of works suggested hybridizing GAs with efficient local search algorithms capable of improving solutions between each generation \parencite[see][Chapter 10]{Eiben2003}. These evolutionary methods have been named in various ways, such as hybrid GAs, memetic GAs, and genetic local search algorithms.

	In the case of $\ICLex$ maximization, existing greedy heuristics may be seen as such local search algorithms, locally improving a partition, and we build on this idea to propose a hybrid GA. In the following, we discuss the practical choices made when designing the genetic algorithm. Moreover, we emphasize that, in this section, the prior parameters are considered to be fixed to uninformative or default values, and we only optimize $\ICLex(\Clust)$ with respect to the partition defined by $\Clust$. 
	


	\subsection{Solution recombination with the cross-partition operator}

	The first step towards building a GA is to define a way to represent candidate solutions inside the algorithm. The latter is also called the \textit{genotype space}, with genotypes as points in this space. This choice is fundamental as it guides the variation operators such as the recombination operator, also known as crossover, which combines two parent genotypes into a new one, and the mutation operator, which randomly modifies genotypes. In the case of clustering, the elements of the original space of solutions are partitions $\Partition = \{\bC_1,...,\bC_{\K}\}$ of $\{1,\hdots, \nb\}$ into $\K$ clusters, with a variable $\K$. \textcite{tessier2006evolutionary} used integer encoding, which consists in a vector of length $\nb$ where each individual is assigned to an integer $k \in \{1 \ldots, \K\}$ representing its cluster assignment, and is similar to $\Clust$. However, this approach presents a major drawback. Indeed, akin to the label switching problem in statistical inference, the $\ICLex$ objective function is invariant under a permutation of the cluster indices, and this representation is therefore heavily redundant. Thus, as emphasized in \textcite{Hruschka2009}, popular crossover operators based on crossover points will not consider this specificity and will completely break the structure of the solution. This is notably the case in \textcite{tessier2006evolutionary}, leading to slow evolution of the population of solutions. We propose to circumvent this issue by directly choosing the space of partitions as the genotype space, defining crossover and mutation operators on it. Such operators will not suffer from label switching, and will preserve the clustering structure present in the genotypes.
	
	\paragraph{Crossover operator} The crossover operator defines how two parent genotypes $\Partition^1 = \{\bC_1^1,...,\bC^1_{K_1}\}$ and $\Partition^2 =\{\bC_1^2,...,\bC^2_{K_2}\}$ are combined together to form an offspring. We propose to use the cross-partition, defined as the set of all possible intersections between the elements of the two partitions:
	\[
	\mathcal{P}^1\times\mathcal{P}^2 := \left\{\bC^1_k \cap \bC^2_l \,,\, \forall k\in \{1,...,\K_1\}, \forall l \in \{1,...,\K_2 \} \right\}\setminus \left\{\emptyset\right\}.
	\]
	This operator produces a new partition with at most $\K_1 \times \K_2$ clusters, which is a refinement of $\Partition^1$ and $\Partition^2$ in the sense that both parents may be reconstructed using merge operations. It is also the first common ancestor of both $\mathcal{P}^1$ and $\mathcal{P}^2$ in the partition lattice. Hence, its interest is twofold. First, as in the motivating example, if both parent partitions are under-fitted, crossing them allows the algorithm to go backward in the partition lattice, considering finer clustering. Second, it is particularly appropriate for the hybridization with greedy heuristics. Indeed, unnecessary clusters may be created when the crossed solutions are around the best one. Then, a greedy local search based on merge moves may be used to remove these clusters efficiently. 
	
	
	
	\subsection{Selection, mutation and the hybrid algorithm}
	
	The remaining aspects of the genetic algorithm concern the selection procedure and the mutation operator. As the population size $V$ is kept fixed throughout the algorithm, selection defines which parent genotypes are combined together to form offspring. On the basis on numerical experiments, we propose to keep a rank-based selection policy \parencite[see][pp.81-82]{Eiben2003}. In this scheme, the selected genotypes for building the next generation are chosen according to a probability proportional to their rank in terms of $\ICLex$. 
	
	As for the mutation operator, it randomly acts on the elements of a genotype, here the clusters of a partition. Together with the recombination operator, it allows introducing variability in the algorithm allowing for a better exploration. Again, a desirable property is the refinement of a given partition, and a natural mutation to consider is to split a cluster in two new ones at random. Then, local searches consisting in swaps and merges can either undo a poor split or explore new directions. The resulting hybrid greedy algorithm is represented as pseudo-code in Algorithm~\ref{alg:hybrid1}.
	
	\begin{algorithm}[hb!]
		\KwData{population size: $V$, probability of mutation: $pm$, maximum number of generations: $maxgen$, dataset $\Obs$ }
		\KwResult{a partition $\mathcal{P}^*$ } 
		\BlankLine
		Build a population $G=\{\mathcal{P}^1,...,\mathcal{P}^V\}$ of initial solutions using $V$ random initializations and greedy swaps.\\
		$nbgen \leftarrow 1$ \\
		\While{$nbgen < maxgen$ }{
			\BlankLine
			add the best solution $\mathcal{P}^*$ in the population to the new generation $G_n = \{\mathcal{P}^*\}$\\
			sample according to their rank in terms of $\ICL$, $(V - 1)$ pairs of solution in $G$\\ 
			\For{each sampled pairs $(\mathcal{P}^1,\mathcal{P}^2)$ of partitions}{
				build the cross partition $\mathcal{P}$ of $\mathcal{P}^1$ and $\mathcal{P}^2$\\
				$\mathcal{P} = \mathcal{P}^1 \times \mathcal{P}^2$\\
				update $\mathcal{P}$ using greedy merge\\
				\If{$random < pm$}{
					sample a cluster of $\mathcal{P}$ and split it randomly in two
				}
				update $\mathcal{P}$ using greedy swap\\
				add $\mathcal{P}$ to the new generation $G_n = \{G_n,\mathcal{P}\}$
			}
			replace the population by the new generation  $G = G_n$ \\
			$nbgen\leftarrow nbgen+1$
		}
		
		return the best solution $\mathcal{P}^*$ of $G_n$
		\caption{Hybrid genetic algorithm}
		\label{alg:hybrid1}
	\end{algorithm}

	\paragraph{Computational efficiency}
	
	From a computational perspective, the crossover and mutation operator can easily be parallelized since they are independent for each pair of solutions to combine. In addition, while already efficient, this first version was optimized by taking advantage of a special feature of the problem. Indeed, after having formed the crossed partition, one may determine the pairs of clusters $(k,l)$ that have a common parent either in $\mathcal{P}^1$ or $\mathcal{P}^2$: 
	\[ \left\{ (\bC_k, \bC_l) \in \left(\Partition_1 \times \Partition_2 \right)^2: \exists \bC \in \mathcal{P}^1 \cup \mathcal{P}^2, (\bC_k \cap \bC \neq\emptyset) \text{ and } (\bC_l \cap \bC \neq\emptyset) \right\}, \]
	only allowing merge and swap movements between them. This allows gaining a factor $\K$, which can be interesting for a large number of clusters, especially in the first iterations of the algorithm. The rationale behind this restriction is that both initial partitions may be recovered if needed, while the inspection of a non-negligible quantity of merge and swap moves having a low chance of being relevant can be avoided. Moreover, the computational cost of a swap or a merge is model dependent, although the $\ICLex$ expressions for the considered models generally allows for efficient formulas. For example, in the case of the binary SBM, \textcite[Appendix B and C]{Come2015} derived swaps and merge updates with $\mathcal{O}(l + \K^2)$ and $\mathcal{O}(\K)$ costs respectively, where $l$ is the average number of edges per node.
	
	\paragraph{Setting genetic hyperparameters} The size of the population of solutions $V$, the probability of mutation and the maximum number of generations are hyperparameters to be set beforehand. They are used to tune the trade-off between computational efficiency and exploration capacity. The hybridization with local maximization methods enhances the exploration capability, therefore reducing the need for large values of these parameters. In \Cref{sec:exp} and in the package, we typically set the population size around $50$, the probability of mutation to $0.25$ and the maximum number of generations to $10$. 
	
	\paragraph{}
	This hybrid genetic algorithm allows the extraction of a natural clustering when the number of clusters is unknown, by carefully exploring the space of partitions and exploiting relevant solutions. The trade-off between the two is controlled by a few tuning parameters, namely the population size and the probability of mutation, and the computational complexity, which is model-dependent, is competitive with other approaches. The experiments carried in Section~\ref{sec:exp} will demonstrate its performances in real and simulated settings. \revision{In practice, relying on a greedy optimization algorithm and allowing merges between clusters leaves a solution with $K>1$ clusters. In the following, we propose a regularization strategy based on the parameter $\alpha$ to allow extra merges and to build a complete hierarchy of clusters up to $K=1$. We emphasize that $\alpha$ is key to unfold complete hierarchies.} 

	\section{Hierarchical extension from regularization path}
	\label{sec:Hierarchical}
So far, we have seen how to maximize the $\ICLex$ with respect to $\Clust$, the prior hyper-parameters being kept fixed, leaving a clustering result at a given level $\K$. We now consider the problem of building a complete hierarchy of clusters using the same criterion. In this section, we introduce the second contribution of this paper: a greedy agglomerative algorithm for hierarchical clustering, based on an approximation of $\ICLex$. Hereafter, $\ICLex$ is viewed not only as a function of the partition $\Clust$ but also of the hyper-parameter $\alpha$. The asymptotic behavior of the $\log$-Gamma function near $0$ is used to derive a simple functional form for the criterion as a function of $\alpha$. The resulting approximation is called $\ICLlin$ due to its log-linear dependency in $\alpha$.  Then, $\alpha$ is used as a  regularization parameter which unlocks access to simpler, coarser, solutions. The algorithm produces a hierarchy of nested partitions along with the sequence of the regularization parameters which enabled the fusions : $( \Clust^{(k)}, \,\alpha^{(k)}  )_{k=K,\ldots,1}$. Eventually, the extracted partitions may be investigated, and the hierarchical structure employed to get a pseudo-ordering of the initial clusters to enhance the graphical representation of the clustering results. \revision{As we shall see, the key advantage of $\ICLlin$ compared to $\ICLex$ is that it allows to obtain explicit values for $\alpha$ allowing merges without having to rely on prohibitive non-linear equations resolution or grid search strategies.}
	
	\subsection{A new approximation for the exact ICL}
	As shown in Equation \eqref{eq:ICLexdec}, $\ICLex$ decomposes as the sum of two terms. The first one is $\log \p(\Obs \mid \Clust,  \bBeta)$, the conditional integrated $\log$-likelihood of the data, given the partition $\Clust$. In the following, it will be denoted by $D(\Clust)$. This quantity only depends on the observed data $\Obs$, the partition $\Clust$, and the model specification. The second term is the integrated $\log$-likelihood of $\Clust$ and depends on the Dirichlet hyper-parameter $\alpha$:
\begin{align}
\label{eq:logpZterm}
\log \p(\Clust \mid \alpha , K) = \log\Gamma(\alpha \K) + \sum_{k=1}^K \log \Gamma(\alpha + n_k) - K \log \Gamma(\alpha) - \log \Gamma(\nb + \alpha\, K) .
\end{align}
Here, the dependency between $\K$ and $\Clust$ is made explicit, the former representing the number of clusters in the latter. Then, we consider the asymptotic behavior of the expression above when $\alpha$ becomes small. First, recall that the $\log$-Gamma function behaves as minus the natural logarithm near $0$:
\begin{equation}
\label{eq:approxGamma}
\log \Gamma(\alpha) = \log(\alpha^{-1} \,\Gamma(\alpha + 1)) \; \approx_{0} \; - \log(\alpha) \,.
\end{equation}
Then, considering $\K$ bounded, we can use this approximation on $\log \Gamma(\alpha)$ and $\log \Gamma(\alpha \K)$ respectively.  Finally, we use the two additional approximations $\log \Gamma (n_k + \alpha) \approx \log \Gamma(n_k)$ and $\log \Gamma(\nb + \alpha \K) \approx \log \Gamma(\nb)$ when alpha is close to $0$. Combining these approximations, a simpler expression of \Cref{eq:logpZterm} as a $\log$-linear function of $\alpha$, can be derived:
\begin{align}
\log \p(\Clust \mid \alpha , K)  & \approx_0 (K-1)\log(\alpha)-\log(K)+\sum_{k=1}^K \log \Gamma(n_k) - \log\Gamma(\nb) \,. \nonumber\label{eq:CondLlhoodApprox}
\end{align}
The algorithm introduced in this paper relies on this approximation, and the corresponding criterion is named $\ICLlin$, where \textit{lin} stands for linear:
\begin{equation}
\label{eq:ICLlin}
\ICLlin(\Clust, \alpha) =   D(\Clust) + (K-1)\log(\alpha) -\log(K) + \sum_{k=1}^K \log \Gamma(n_k) - \log \Gamma(\nb)\nonumber \, .
\end{equation}
All quantities that do not depend on $\alpha$ may be grouped in an intercept:
\begin{equation}
\label{eq:intercept}
I(\Clust) \coloneqq  D(\Clust) - \log(K) + \sum_{k=1}^K \log \Gamma(n_k) - \log \Gamma(\nb) \, .
\end{equation}
\revision{Note that the $\Gamma(n_k)$ term is always well-defined here, since we do not consider empty clusters.} Then, the $\log$-linearity of the new criterion appears explicitly:
\begin{equation}
\label{eq:ICLlinclear}
\ICLlin(\Clust, \alpha) = (K-1) \log(\alpha) + I(\Clust) \,.
\end{equation}
Naturally, the quality of this approximation depends on how small both $\alpha$ and $\alpha \K$ are. For the first one, the approximation of \Cref{eq:approxGamma} is quite mild, even for standard $\alpha$ values such as $1$ or $\frac{1}{2}$. As for $\alpha \K$, while its value may be relatively far from $0$ for $\alpha = 1$, we verify in practice that $\alpha$ rapidly decreases several orders of magnitude below $1$ as of the first fusion. This ensures that the approximation is correct throughout the procedure.

	\subsection{Hierarchy construction}

	Looking at the functional form of the previous approximation, a natural goal is to search for the Pareto front in the $(\log \alpha, \,\ICLlin(\Clust,\alpha))$ plane. The latter corresponds to a set of dominating partitions with respect to $\ICLlin$, for a certain range of $\alpha$ values in $]0,1]$, or equivalently for a range of $\log(\alpha)$ values in $]-\infty,0]$. Formally, we define the Pareto front as: 
	\begin{align}
	\label{eq:pareto}
	P  = \{ (\Clust^\star, I^\star_{\alpha}) \; : \; \forall \alpha \in I^\star_{\alpha}, \, \forall \Clust \neq \Clust^\star, \, \ICLlin(\Clust^\star,\alpha) \geq \ICLlin(\Clust,\alpha) \},  
	\end{align}
	where $I^\star_{\alpha}$ are intervals of $]0,1]$. Finding this set of dominating partitions and ranges is not a trivial task. However, the difficulty is reduced if we consider a dominant partition $\Clust$ for certain level $\alpha$, and restrict ourselves to look for partitions that results from merges of $\Clust$. Indeed, we will show that it is direct, for a given partition $\Clust$, to find the hyper-parameter $\alpha^\star$ and the pair $(g^*,h^*)$ of clusters to merge, such that the obtained coarser partition $\Clust_{g^*\cup h^*}$ will dominate $\Clust$, along with any other partition $\Clust_{g \cup h}$, over $]0,\alpha^*]$. Starting from an initial clustering $\Clust^{(K)}$, these locally optimal merges can be used to build a heuristic, in the spirit of hierarchical agglomerative clustering, that will extract a sequence of nested partitions to approximate the Pareto front defined by Equation \eqref{eq:pareto}. While this heuristic is not guaranteed to extract the Pareto front, it may still provide good results, especially starting from a dominant partition, \textit{e.g.} obtained by maximizing $\ICLex(\Clust,1)$ with the hybrid optimization algorithm introduced in the previous section. Intuitively, if a partition $\Clust$ is locally dominant for some $\alpha$ value, there is a good chance that the next dominant partition for some $\alpha'<\alpha$ will be a coarse version of $\Clust$. Indeed, to surpass a dominant solution in $\alpha'$, the new dominant solution must be coarser in order to benefit from a reduced decreasing slope, while it must also have a high intercept $I(\Clust')$. Solutions built from merging two clusters of $\Clust$ are coarser, therefore fulfilling the first requirement. Moreover, since $\Clust$ is already dominant, we may also hope that a coarser version of it also has a high intercept, and therefore dominates other partitions for this new $\alpha'$ value. Let us therefore detail this heuristic, and the conditions under which a fusion opportunity exists. 
	
	\subsubsection{Fusion opportunity}
	\label{sub:mergemovement}
	
	For any given partition $\Clust^{(k)}$, with $k \geq 2$ clusters, let us define $\mathbb{Z}^{(k-1)}$ as the space of all the partitions with $(k-1)$ clusters that are coarser than $\Clust^{(k)}$:
	\[
	\mathbb{Z}^{(k-1)} = \left\{ \Clust_{g \cup h} \, : \text{the partition } \Clust^{(k)} \text{ with clusters g and h merged, }  g \neq h \right\} \,.
	\] 
	Note that we will use the terminology \textit{mother} partition for $\Clust^{(k)}$  and \textit{child} partition for any element of $\mathbb{Z}^{(k-1)}$. 
	
	As pointed out previously, with $\Clust^{(k)}$ fixed, the function  $\ICLlin(\Clust^{(k)}, \, \cdot)$ is $\log$-linear with slope $(k - 1)$ and intercept $I(\Clust^{(k)})$. This implies that the slope of the $\ICLlin$ functions decreases incrementally to $0$ as $k$ decreases to $1$. Figure~\ref{fig:droitesICL} illustrates this behavior of $\ICLlin$, with respect to the number of clusters $k$. It can easily be seen that the slopes decrease until $k$ reaches 1 which corresponds to an horizontal line.
	\begin{figure}[!ht]
		\centering
		\includegraphics[width=0.75\textwidth]{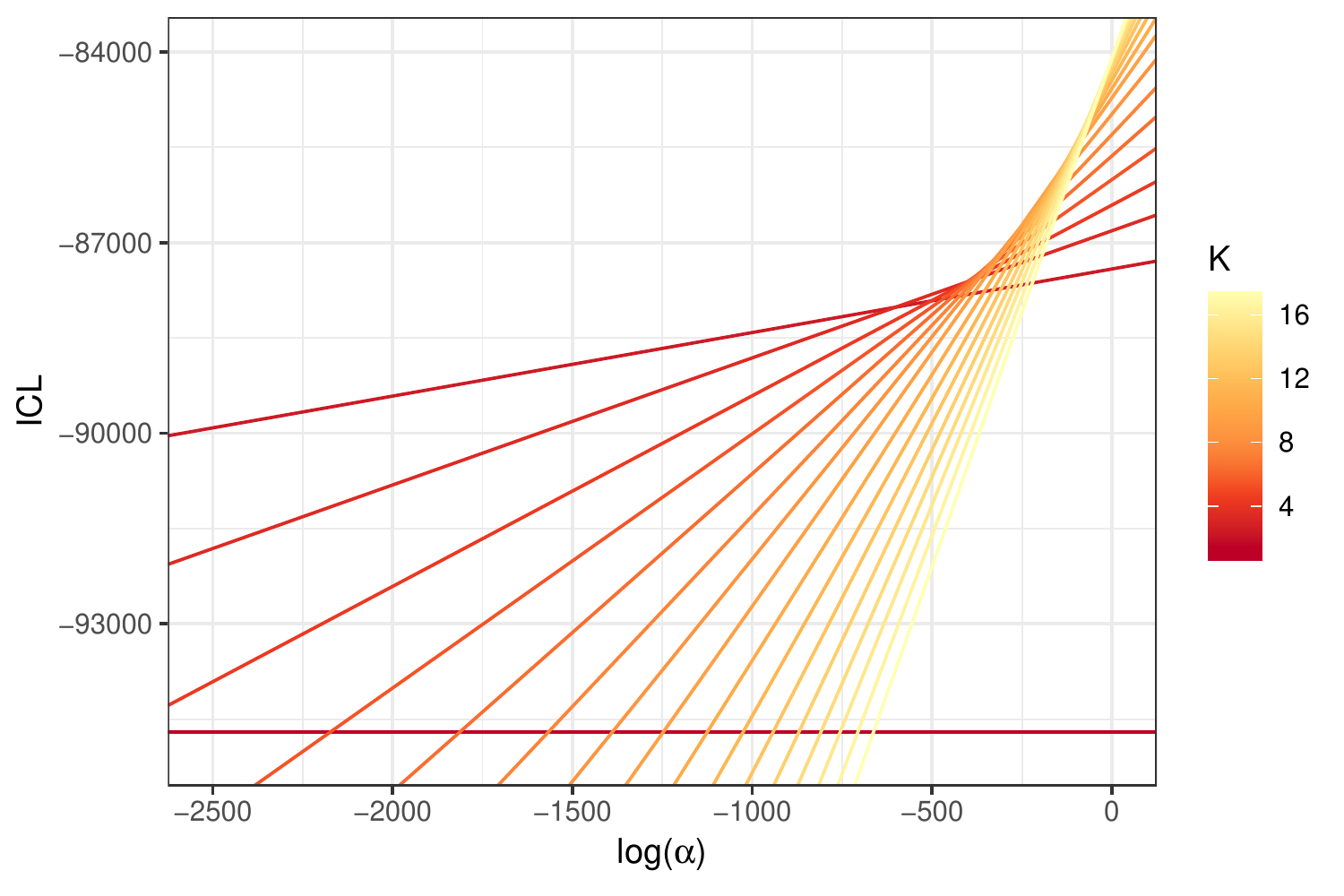}
		\caption{Lines of slope $k-1$ representing the functions $\log\alpha \mapsto ICL(\Clust^{(k)}, \, \log\, \alpha)$ for a collection of partitions $\Clust^{(k)}$ with a decreasing number of clusters $k=21, \ldots, 1$. We see that the $\ICLlin$ order changes as $\alpha$ decreases, favoring coarser partitions. The x-axis slice at $\log \alpha=0$ corresponds to the intercepts $I(\Clust^{(k)})$.}
		\label{fig:droitesICL}
	\end{figure}
	
	\revision{From Equation \eqref{eq:ICLlinclear}, we are able to derive the expression of the variation of the $\ICLlin$ between a mother partition $\Clust^{(k)}$ and any of its child $\Clust_{g \cup h}$ as a function of $\alpha$. Graphically, this variation, denoted as $\Delta_{g\cup h}(\alpha)$, is the difference between two straight lines of slope $k - 2$ and $k - 1$ respectively. Moreover, the dominance shifting point corresponds to its zero, which can be easily derived and will be denoted by $\alpha_{g,h}$}:
	\begin{equation}
	\label{eq:deltaZero}
	\begin{aligned}
		\Delta_{g\cup h}(\alpha_{g,h}) = 0 & \iff \ICLlin\left(\Clust_{g \cup h}, \alpha_{g,\, h}\right) - \ICLlin\left(\Clust^{(k)}, \alpha_{g, h}\right) =0\, , \\
 & \iff  \log(\alpha_{g, h}) \coloneqq I(\Clust_{g \cup h}) - I(\Clust^{(k)})  \,. 
		\end{aligned}
	\end{equation}
	In geometric terms, we know that below this level, the child partition $\Clust_{g \cup h}$ dominates its mother $\Clust^{(k)}$ in terms of $\ICLlin$. Thus, for any mother partition $\Clust^{(k)}$, we are able to compute the tipping points $(\alpha_{g,h})_{g<h}$ for the $k (k - 1)/2$ possible child partitions. To find the best fusion, we recall the form of the $\ICLlin$ for any partition $\Clust_{g\cup h} \in \mathbb{Z}^{(k-1)}$ from Equation \eqref{eq:ICLlinclear}:
	\[
	\ICLlin(\Clust_{g \cup h}, \alpha) = (k-2)\log(\alpha) \; + \; I(\Clust_{g\cup h}) \, , \quad \forall g,h \, .
	\]
	So it is clear that, viewed as functions of $\log  \alpha $, the $\ICLlin$ of all child partitions in $\mathbb{Z}^{(k-1)}$ are parallel straight lines of slopes $(k-2)$, only differing by their intercepts. This guarantees us that there exists a unique partition, uniformly dominating in $\alpha$, in $\mathbb{Z}^{(k-1)}$. Formally:
	\begin{align}
	\exists ! \Clust_{g^{\star} \cup\, h^{\star}} \in \mathbb{Z}^{(k-1)}  \text{ s.t. : } \forall \alpha > 0 , \, \forall \Clust_{g \cup h} \in \mathbb{Z}^{(k-1)} \nonumber \\
	\ICLlin\left( \Clust_{g^{\star} \cup\, h^{\star}}, \alpha \right) &\geq \ICLlin\left(  \Clust_{g \cup h}, \alpha \right) \, . \label{eq:uniqueZ}
	\end{align}
	This partition corresponds to the one with the greatest intercept which, by Equation \eqref{eq:deltaZero}, also happens to be the one intersecting with $\Clust^{(k)}$ at the greatest $\alpha_{g,h}$: \\
	\[
	(g^{\star}, h^{\star}) = \arg\max\limits_{g,h} I(\Clust_{g\cup h}) = \arg\max\limits_{g,h} I(\Clust_{g\cup h}) - I(\Clust^{(k)}) = \arg\max\limits_{g,h} \alpha_{g,h}.
	\]
	This discussion describes how to find the best fusion, going from a partition $\Clust^{(k)}$ to $\Clust^{(k-1)} = \Clust_{g^{\star} \cup\, h^{\star}}$ by setting $\alpha^{(k-1)} = \alpha_{g^{\star}, h^{\star}}$. Taking this greedy approach, one may perform such locally optimal merges sequentially in a fast and efficient bottom-up procedure until all clusters have been merged into a unique cluster. Hence, we can see how $\alpha$ acts as a regularization parameter, enabling for fusions. Taking an initial partition $\Clust^{(K)}$ and a given initial $\alpha^{(K)}$, typically~$1$, this will provide a set of nested clustering solutions $(\Clust^{(k)}, \alpha^{(k)})_{k=K,\ldots,1}$. 
	
	Finally, the log-scale used for the $x$-axis of \Cref{fig:droitesICL} highlights that $\alpha$ quickly decreases toward $0$ from the first iteration, thus insuring the validity of the approximation at the first stages. This fact is also empirically verified in all the experiments of \Cref{sec:exp}, and is easily visualized via the dendrogram representation introduced below.
	
	\subsubsection{Post-processing}
	
	The previous strategy outputs a hierarchy, meaning a set of nested clustering with a number of clusters ranging from $K$ to $1$. Each merge performed by the algorithm is stored into a binary tree, keeping track of the hierarchical relations between clusters. However, one important point to observe is that some of the partitions extracted by this agglomerative greedy algorithm may not be dominant anywhere in $\alpha\in]0,1]$, with respect to the others. This corresponds to situations where combining several merges in one step is better than performing them sequentially. Indeed, in geometrical terms, there is no guarantee that the intersection between the $\ICLlin$ of $\Clust^{(k)}$ and $\Clust^{(k-1)}$ is at a greater $\alpha$ than between $\Clust^{(k)}$ and $\Clust^{(k-2)}$. Or, equivalently, there is no guarantee that the sequence $(\alpha^{(k)})_k$ is non-increasing. This is quite natural since $\ICLlin$ is a penalized criterion, thus it does not necessarily increase with the model complexity. Since such partitions cannot belong to the approximated Pareto front, we propose to remove them.  Indeed, they are easy to track since they correspond to merge $k$ where $\alpha^{(k-1)} > \alpha^{(k)}$. Then, having extracted the $F \leq K$ dominating partitions, it is possible to recompute the $\alpha_f$ where they cross each other to get a sequence $(\Clust_f,\alpha_f)_{f=F,...,1}$ with a non-increasing sequence $(\alpha_f)_f$. Although the index of $\Clust_f$ does not indicate its number of clusters anymore, the sequence still consists in a hierarchy of nested partitions, which are now ordered in terms of $\ICLlin$ in their ranges of dominance: $\ICLlin(\Clust_f,\alpha) > \ICLlin(\Clust_{l},\alpha),\,\forall\,f \neq l,\,\forall \alpha \in[\alpha_{f-1},\alpha_{f}]$. Figure \ref{fig:front} illustrates this post-processing, where the $\ICLlin$ lines associated with each $\Clust^{(k)}$ extracted by the greedy agglomerative algorithm are depicted with their corresponding dominance ranges, and the nowhere dominant partitions are highlighted.
	\begin{figure}[!ht]
		\centering
		\includegraphics[width=0.75\textwidth]{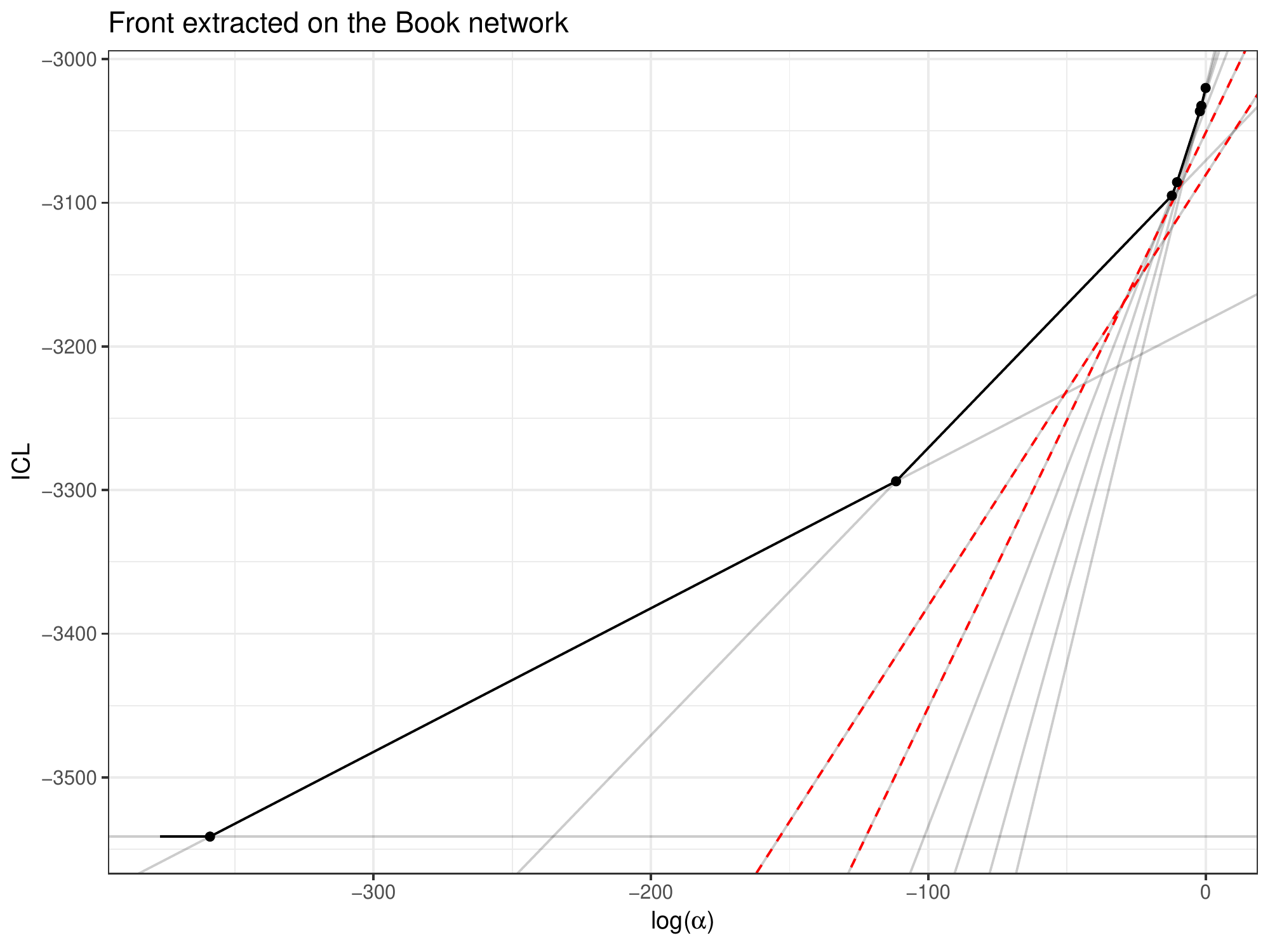}
		\caption{$\ICLlin(\Clust,\alpha)$ as a function of $\log(\alpha)$ for every partition extracted by the greedy hierarchical algorithm on the \textbf{Books} co-purchasing network (see Section~\ref{sub:RealDataHybrid} for dataset details), with a degree corrected SBM model. The partitions that do not have any range of dominance are highlighted with dashed red lines, and the dominant ranges with solid black lines. The intersection between dominant partitions correspond to the recomputed tipping, dominance shifting points $(\alpha_f)_f$. The initial partition $\Clust^{(K)}$ was built using Algorithm \ref{alg:hybrid1}.}
		\label{fig:front}
	\end{figure}
	
	\subsubsection{Visualization}
	Along with its property discussed above, the proposed algorithm possesses interesting graphical features for the visualization of both the hierarchy, with a dendrogram, as well as the initial clustering $\Clust^{(K)}$ using the partial ordering of the leaves.
	
	\paragraph{Dendrogram} The sequence $(\alpha_{f})_{f=\, F,\,  \cdots \, ,\, 1}$ may be used for the construction of a dendrogram representing the cluster merge tree from $\Clust^{(K)}$ to $\Clust^{(1)}$, with the non-increasing sequence $(-\log(\alpha_f))_f$ in the $y$-axis. Thus, the hierarchical structures of the clusters can be visualized as well as the amount of regularization needed for each fusion(s). Indeed, as discussed above, the $y$-axis can then be seen as the drop in $\ICLlin$ induced by each merge, acting as an analog of the traditional \textit{dissimilarity} in agglomerative strategies. Figure \ref{fig:dendo} presents the obtained dendrogram for the Book network of Section \ref{sub:RealDataHybrid}. 
	\begin{figure}[!ht]
		\centering
		\includegraphics[width=0.75\textwidth]{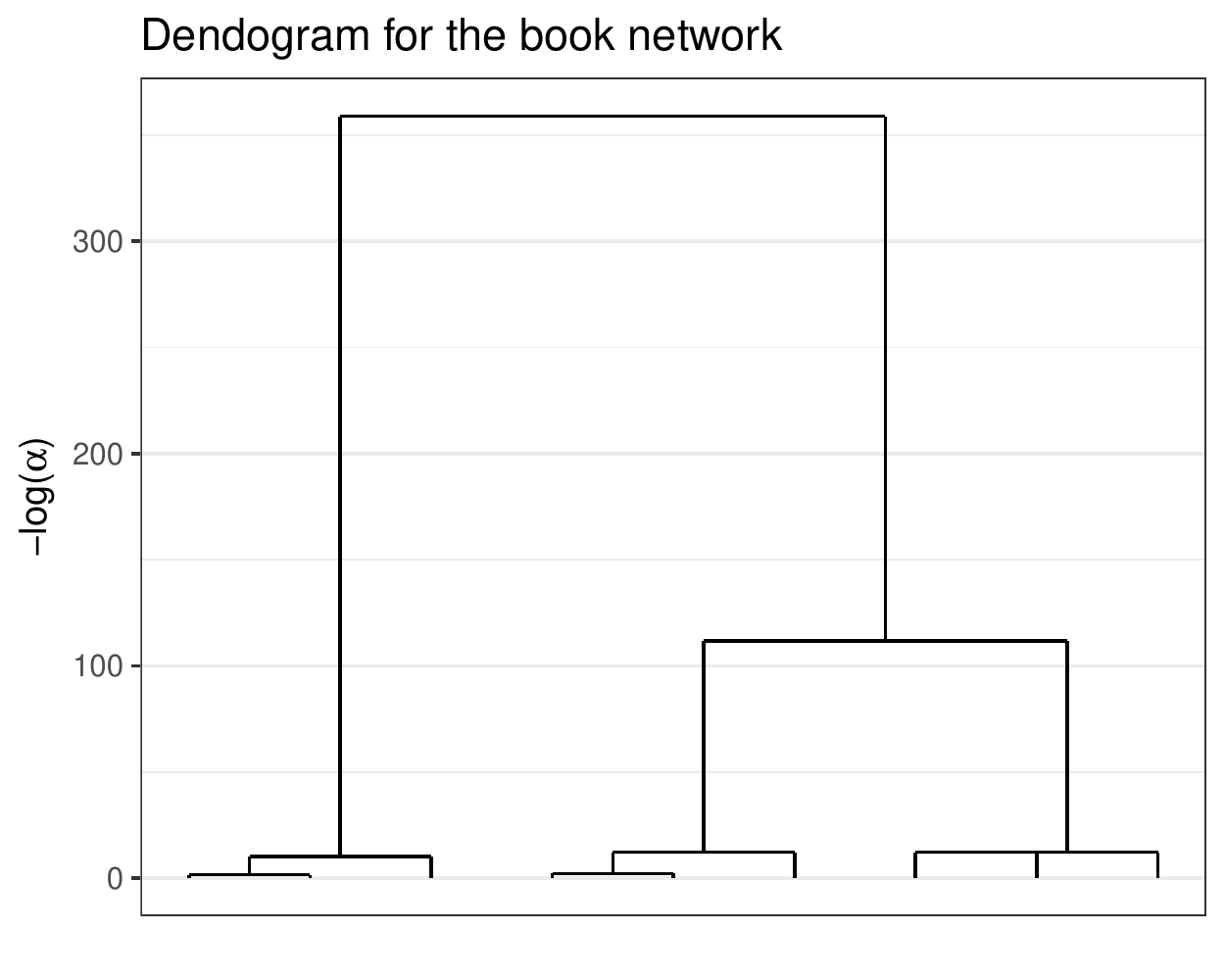}
		\caption{Dendrogram representation of the extracted hierarchy for the \textbf{Books} co-purchasing network (see Section~\ref{sub:RealDataHybrid} for dataset details).}
		\label{fig:dendo}
	\end{figure}
	
	\paragraph{Leaves ordering} Another interesting feature of the proposed procedure is the partial ordering of the initial clustering $\Clust^{(K)}$ that can be obtained from the merge tree structure. Indeed, for a binary tree with $\K$ leaves, there are $2^{K-1}$ permutations of its leaves that are compatible with its structure. In other words, there are $2^{K-1}$ possible dendrograms representing the same hierarchy. However, some are more relevant than others and we seek to find the optimal tree consistent ordering (or permutation) $\sigma$ that minimizes the sum of merge costs between successive clusters at $\alpha=1$:
	\begin{equation}
	\sigma =\arg\min_\sigma \sum_{k=1}^{K-1}\Delta_{\sigma(k) \cup \sigma(k+1)}.
	\end{equation}
	An efficient algorithm based on dynamic programming \parencite{Bar2001} is already available to solve this optimization problem using the binary tree structure. As shown in Figure \ref{fig:motiv}, such ordering of the initial clusters may be used advantageously to draw node-link diagrams or block adjacency matrix, enhancing visualization and simplifying the interpretation of the clustering results. This approach is used in the \pkg{greed} package to provide the final ordering of the clusters.

	\section{Numerical experiments}
	\label{sec:exp}
    Thus far, the discussion has been purposely general in order to express the generic aspect of the proposed two-fold methodology. We now illustrate
	the behavior of the two algorithms in simulated and real settings for several particular instances of DLVMs. First, the hybrid optimization algorithm is compared with other clustering algorithms in simulated scenario for count-data and network clustering. The results highlight both its improvement of hill climbing heuristics, as well as its advantage compared to standard statistical approaches. The results of the hierarchical algorithm is then analyzed on several real datasets for graph clustering and co-clustering, demonstrating its interest in finding relevant hierarchical structures. The formula of $\ICLex$ are derived in \cref{sec:models} for each specific model.
	
	\subsection{Medium-scale SBM simulations}
	
	To investigate the performances of the hybrid algorithm, we pursue with our motivating example defined in Section~\ref{HCICL:intro}. The simulation consists of a SBM graph with 1500 nodes and 15 clusters with a hierarchical structure of 3 clusters each divided into 5 small clusters. Figure~\ref{fig:sbmsim1} (left) presents the evolution of the $\ICL$ criterion among the different generations of solutions built by the algorithm. As clearly shown by this figure, the criterion improves at each generation until it reaches a plateau around the fourth generation. 
	\begin{figure}[!ht]
		\centering
		\includegraphics[width=0.95\textwidth]{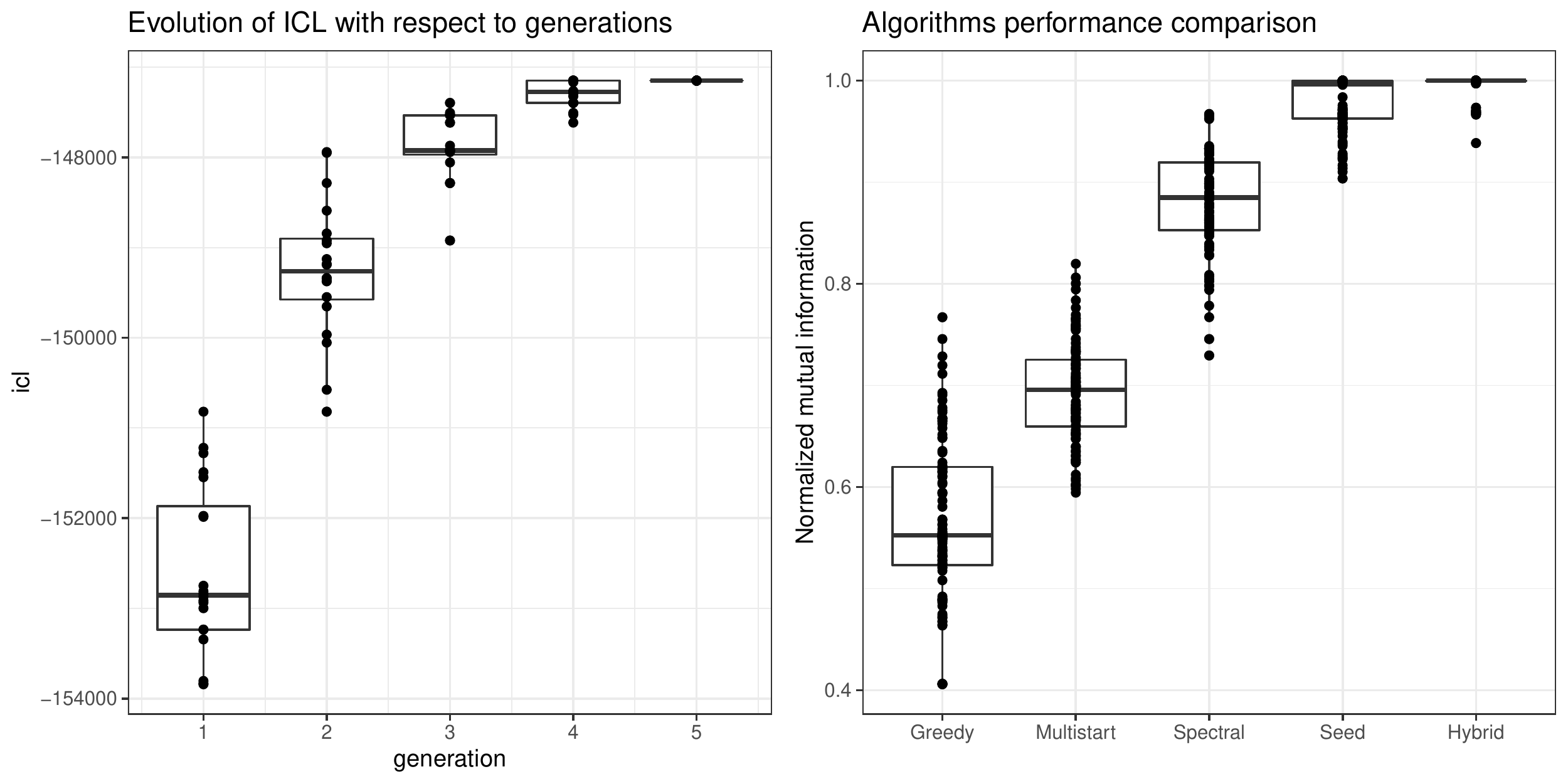}
		\caption{Evolution of $\ICLex$ with respect to the generation for one run of the hybrid algorithm (left), NMI between simulated and reconstructed clusters for one hundred simulations for the different algorithms (right).}
		\label{fig:sbmsim1}
	\end{figure}
	A comparison of the hybrid algorithm with other methods is also performed on the same problem by running the different algorithms with one hundred simulated graphs. The hybrid algorithm is compared with a greedy algorithm with random starting point, a greedy algorithm with multiple random starting partitions, a regularized spectral algorithm \parencite{Qin2013}, and a greedy algorithm initialized with the spectral algorithm. All the methods include model selection in their core, except for the spectral clustering which is thus advantageously run with the true number of clusters. For all the variants of the greedy algorithm and our hybrid proposal default values were used for their parameters: initial number of clusters equal to twenty, size of the population equal to fifty, probability of mutation equal to 0.25 and maximum number of generations fixed to ten.
	The comparison is made in terms of normalized mutual information \parencite[NMI,][]{vinh2010information} between the extracted and simulated clusters. The NMI allows comparing partitions with a different number of clusters, as is needed in this setting, and an NMI of 1 means a perfect match between two partitions. As expected, the greedy algorithm with random starting point suffers from quite severe under-fitting and gives an NMI around 0.55, using multistart helps a little and the solutions then are around an NMI of 0.7. The spectral algorithm does also improve with an NMI around 0.85. Eventually, the two best algorithms are the simple greedy algorithm carefully initialized (here using the results of the spectral algorithm with twenty clusters) and our proposed hybrid algorithm which recovers almost perfectly the simulated partitions in all of the simulations (93\% of perfect recovery) whereas some simulations are still not perfectly recovered by the greedy algorithm with careful initialization (51\% of perfect recovery).  
	
	\subsection{Medium-scale mixture of multinomials simulations}
	
	As a second scenario, we focused on a mixture of multinomials. The simulation setup was as follows: 15 clusters with equal proportions were generated. The sample size was fixed to 500 and the number of possible outcomes for the multinomials to 100. The multinomial parameters were set such that each cluster has a uniform distribution on $\{1, \ldots, 100\}$ except for 10 randomly chosen outcomes that have their probabilities multiplied by 4. Eventually the number of draws for each multinomial sample was set to 50. The simulation was performed one hundred times and for each generated dataset the solutions found by the different variants of the greedy heuristic, an EM algorithm (from the \pkg{mixtools} R package) with model selection performed with AIC and BIC were recorded. We may first look at the number of clusters extracted by each algorithm. Figure \ref{fig:mmsim1} presents the bar graphs of the number of extracted clusters for each of the algorithms over the 100 generated datasets.
	
	\begin{figure}[!htp]
		\centering
		\includegraphics[width=0.95\textwidth]{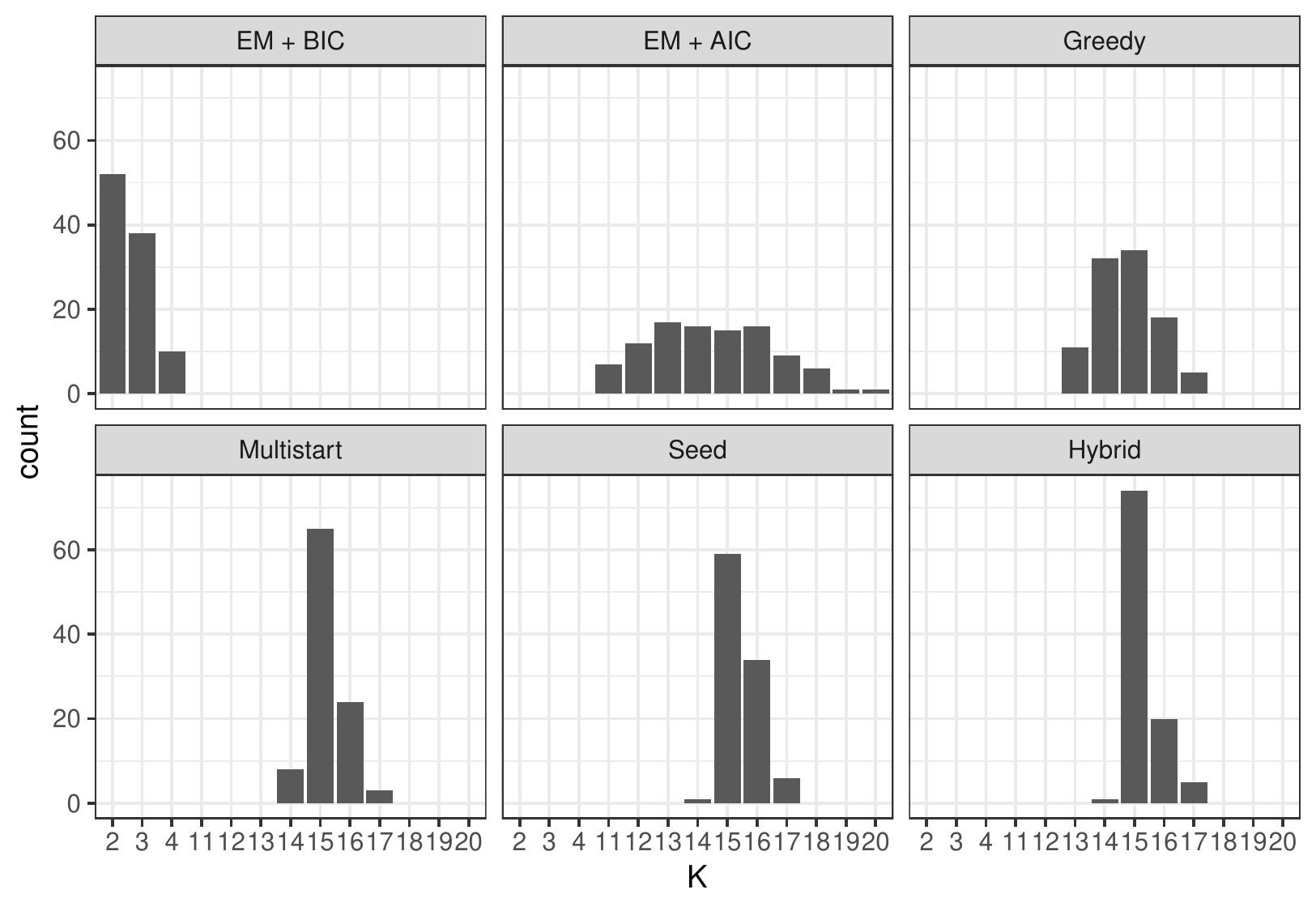}
		\caption{Bar graphs of the number of  extracted clusters over one hundred simulated datasets for the different algorithms. The datasets were generated with $K=15$.}
		\label{fig:mmsim1}
	\end{figure}
	
	The solutions found using an EM algorithm and BIC or AIC for model selection suffer from a lot of variance. AIC gives more satisfactory results on this problem but the number of extracted clusters is still quite variable, between 10 and 22. BIC leads to too simple models with fewer than 5 clusters in all the simulations. Some of these results can be explained by the random initialization of the EM algorithm. Greedy maximization of $\ICL$ gives better results in this problem and found the correct number of clusters in around 60\% of the simulations with the multistart version of the algorithm (which is a little bit better than the version seeded with a simple $k$-means). Eventually, the hybrid algorithm found the correct number of clusters in more than 75\% of the simulations and is therefore also better here.
	If we inspect the results with respect to the NMI with the simulated labels, or with the obtained ICL values as shown in Figure \ref{fig:mmsim2}, the ranking of the different solutions does not differ. The hybrid algorithm leads to the best results even though the differences with the seeded version of the greedy algorithm are less important with respect to these metrics in this experiment.
	
	\begin{figure}[!htp]
		\centering
		\includegraphics[width=0.95\textwidth]{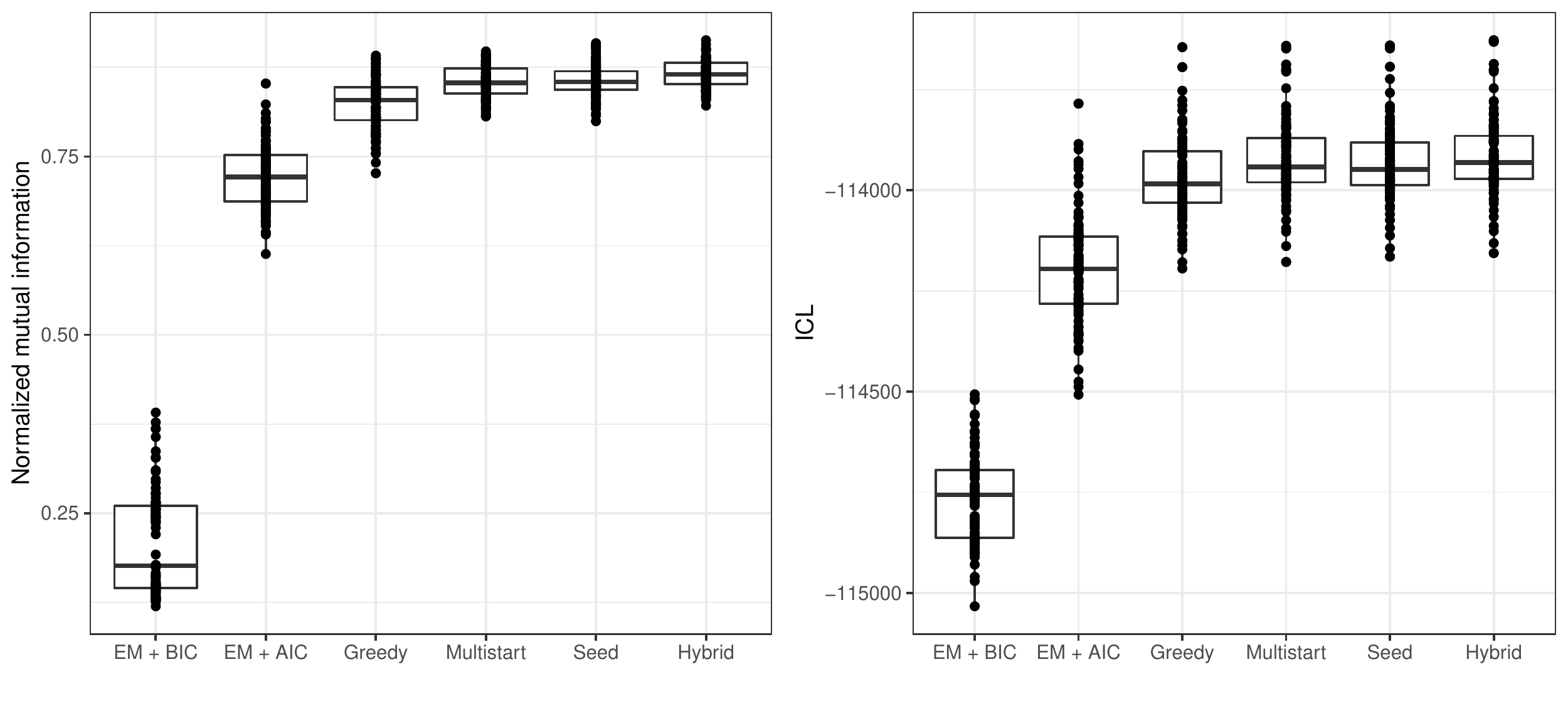}
		\caption{NMI between simulated and extracted clusters and ICL for the different algorithms on the mixture of multinomial simulation over one hundred simulations.}
		\label{fig:mmsim2}
	\end{figure}
	
	\subsection{Clustering real network data}
	\label{sub:RealDataHybrid}
	
	The performances of the proposed solution were also investigated with real datasets. Classical graph clustering datasets were first analyzed:
	\begin{itemize}
		\item\textbf{Blog}: a directed network from \textcite{Adamic2005} of hyperlinks between 1222 blogs on US politics, recorded during the 2004 presidential election,
		\item\textbf{Books}: a network of 105 books about US politics also published around the time of the 2004 presidential election and sold by the online bookseller Amazon.com (edges between books represent frequent co-purchasing of books by the same buyers),
		\item\textbf{Jazz}: an undirected network of 198 jazz bands \parencite{Gleiser2003},
		\item\textbf{Football}: an undirected network of American football games between 115 colleges during the regular Fall 2000 season \parencite{Newman2004}.
	\end{itemize}
	All of these classical datasets were downloaded from Mark Newman datasets page\footnote{available at \href{http://www-personal.umich.edu/~mejn/netdata/}{http://www-personal.umich.edu/~mejn/netdata/}}. Two co-clustering datasets were also benchmarked: 
	\begin{itemize}
		\item \textbf{French parliament}: this dataset concerns the votes of 593 French deputies during a part of the current legislature and covers 1839 ballots, the data were extracted from the French national assembly open data api\footnote{available at \href{http://data.assemblee-nationale.fr/}{http://data.assemblee-nationale.fr/}} and gathered into a binary matrix where the presence of a one indicates a positive vote of a deputy for a specific ballot. \item \textbf{Jazz bands / musicians}: is a recreation of the raw data in \textcite{Gleiser2003}. These raw data were extracted by scrapping the same source namely \textit{The Red Hot Jazz Archive}\footnote{available at \href{http://www.redhotjazz.com/}{http://www.redhotjazz.com/}}. For each available band, the list of its members was extracted leading to a binary matrix of 4475 musicians and 965 bands. For all the performed analyses,  we removed all the musicians that played in fewer than 3 bands and all the bands with fewer than 3 musicians, leaving a final matrix of 690 musicians and 539 bands.
	\end{itemize}
	
	These original datasets were produced for this paper and are available together with the classical network datasets in the \pkg{R} package 
	accompanying the paper. For each of these datasets, and in order to get some information on the variability of results, we ran the algorithms 25 times, with a degree corrected SBM (dc-SBM) model for networks and degree corrected LBM (dc-LBM) model for co-clustering datasets, and the resulting $\ICLex$ values were recorded. The algorithms are the same as previously: greedy with multiple random starts, seeded greedy (spectral algorithm for dc-SBM and independent $k$-means on rows and columns for dc-LBM) and our proposed hybrid approach. To study the impact of the population size on the results of the hybrid algorithm, this parameter was also set to vary in $\{20,40,80\}$. These numbers are quite small with respect to the ones commonly encountered in pure GA, which is allowed by the use of hybridization with local search reducing the need for a large population.

	\begin{figure}[!htp]
		\centering
		\begin{tabular}{cc}
			\includegraphics[width=0.49\textwidth]{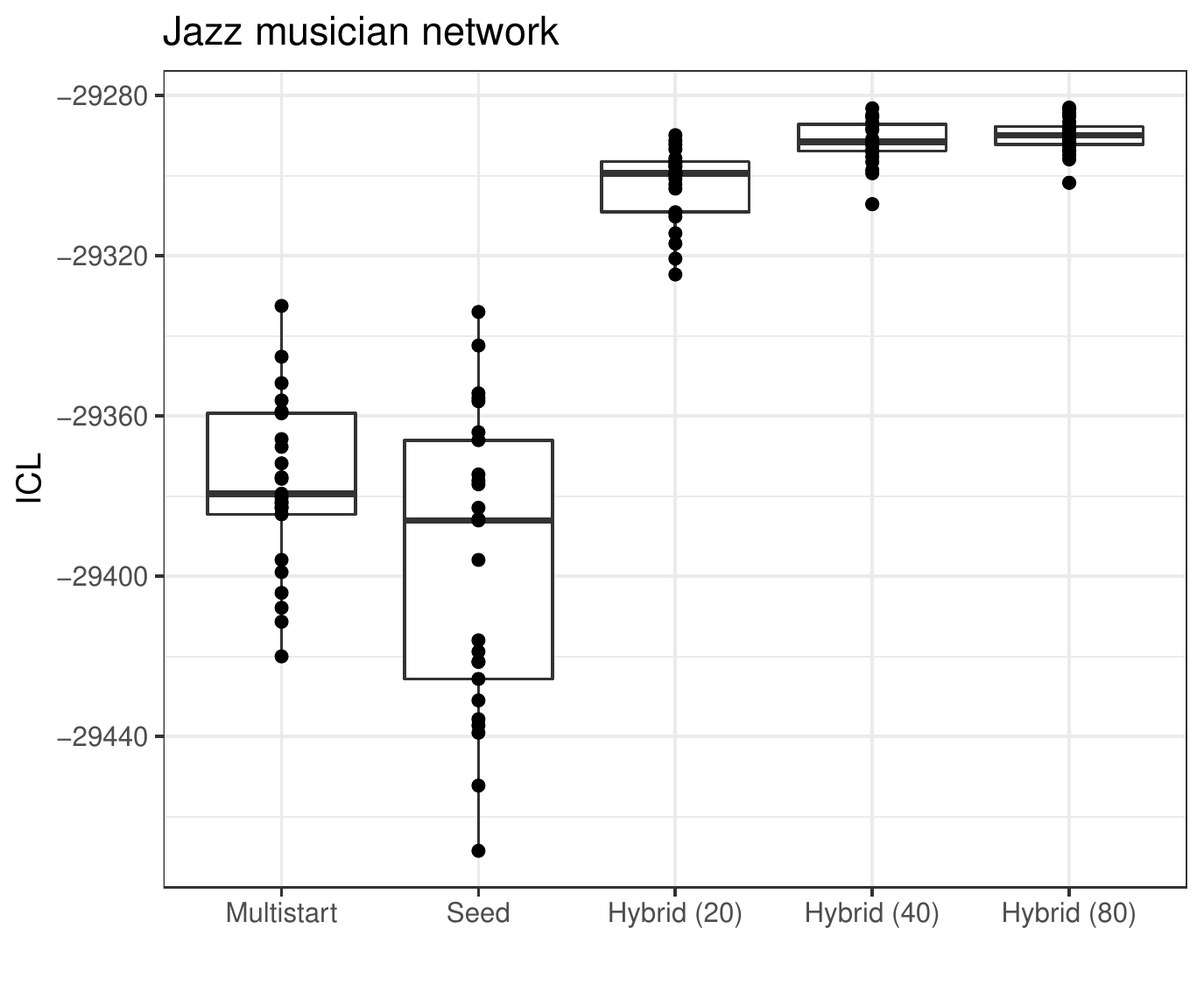}&\includegraphics[width=0.49\textwidth]{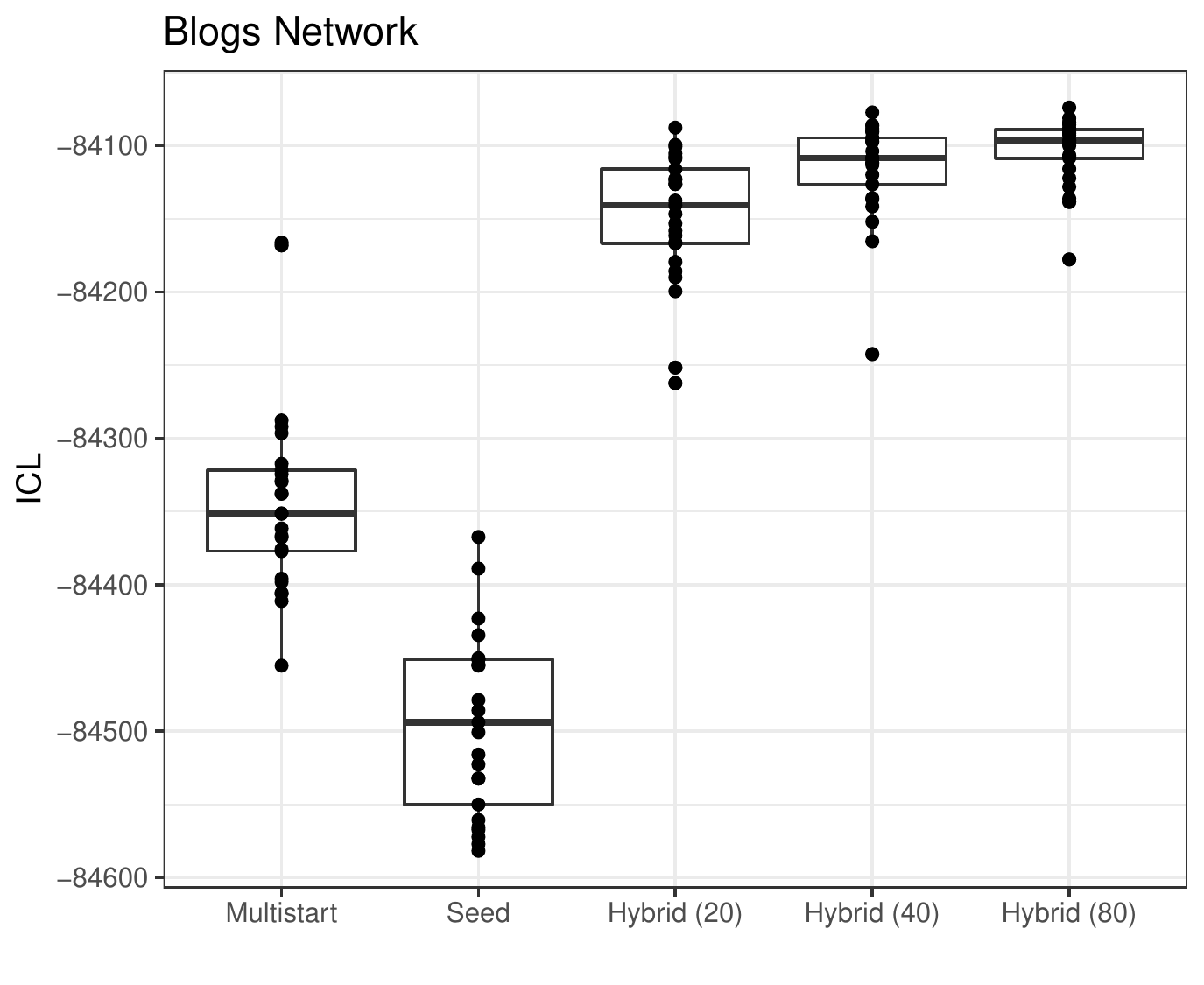}\\
			\includegraphics[width=0.49\textwidth]{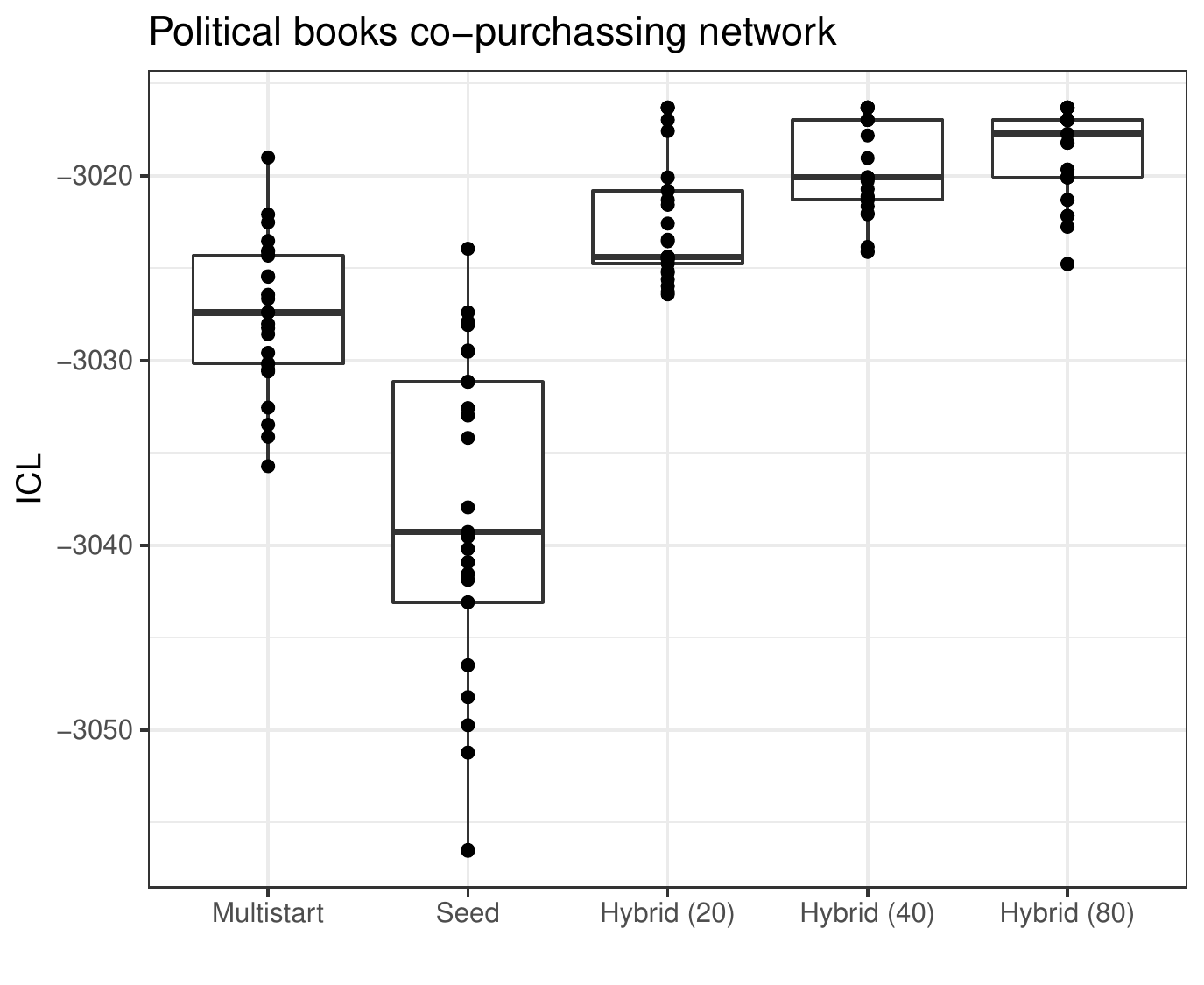}&\includegraphics[width=0.49\textwidth]{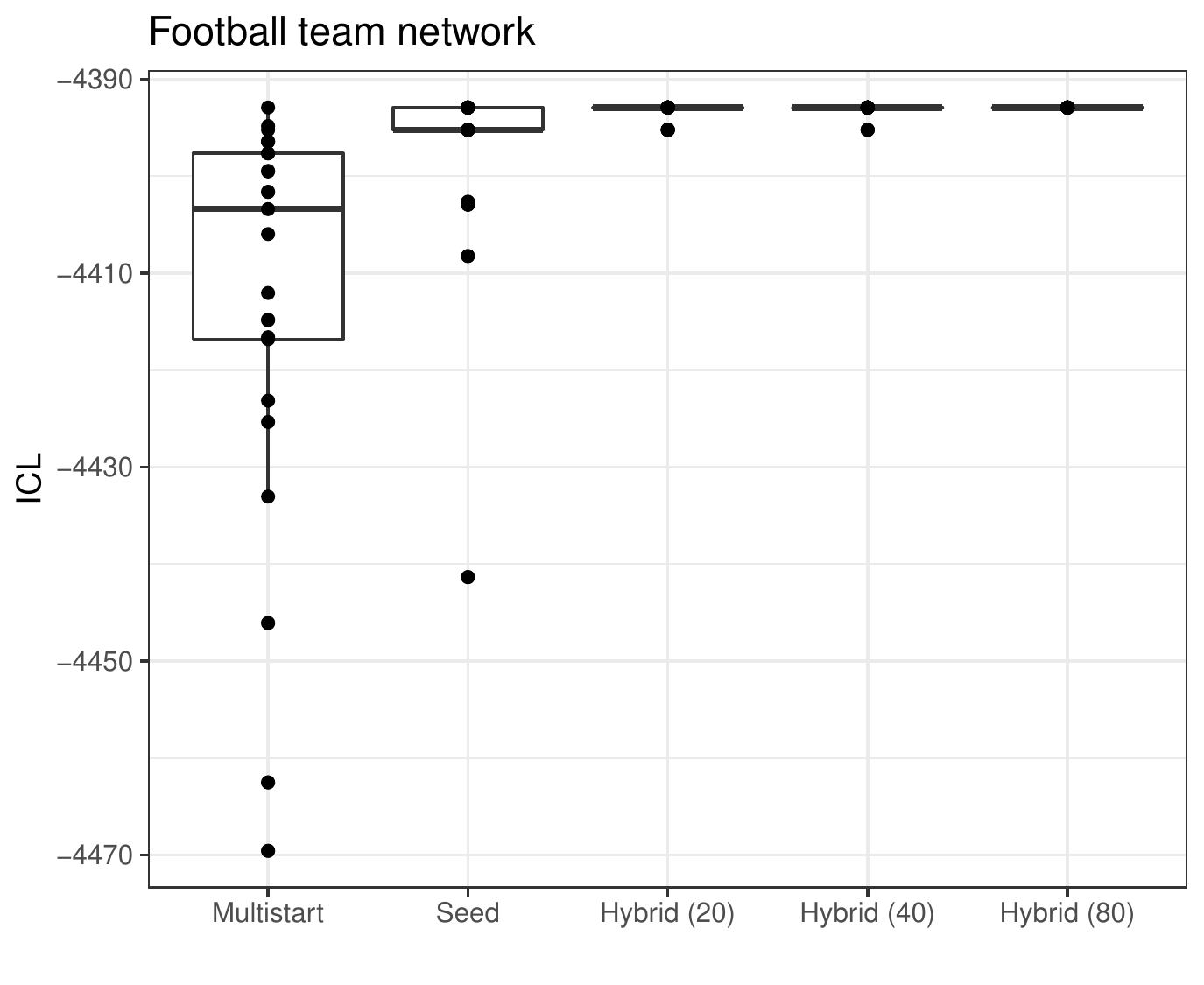}\\
			\includegraphics[width=0.49\textwidth]{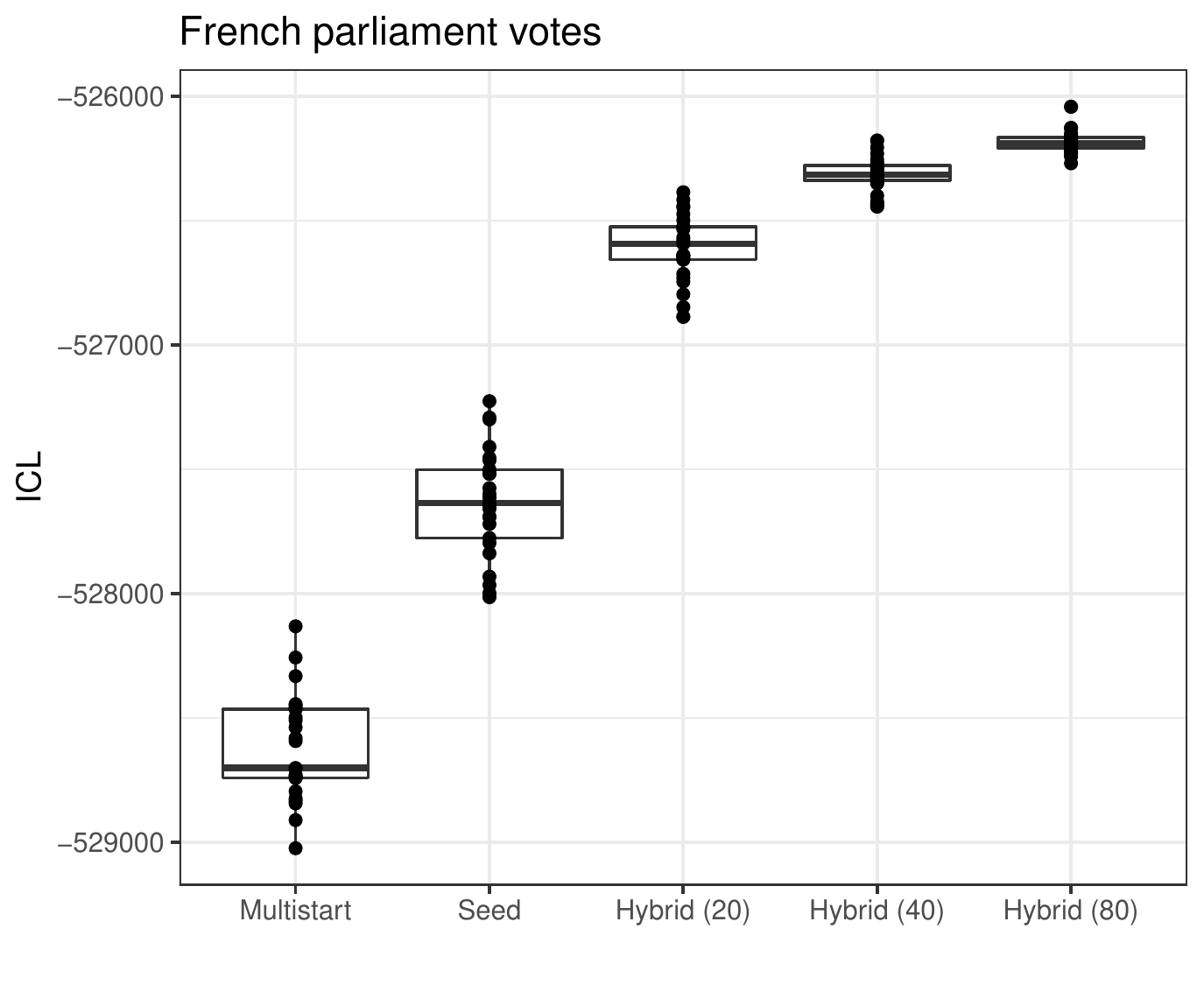}&\includegraphics[width=0.49\textwidth]{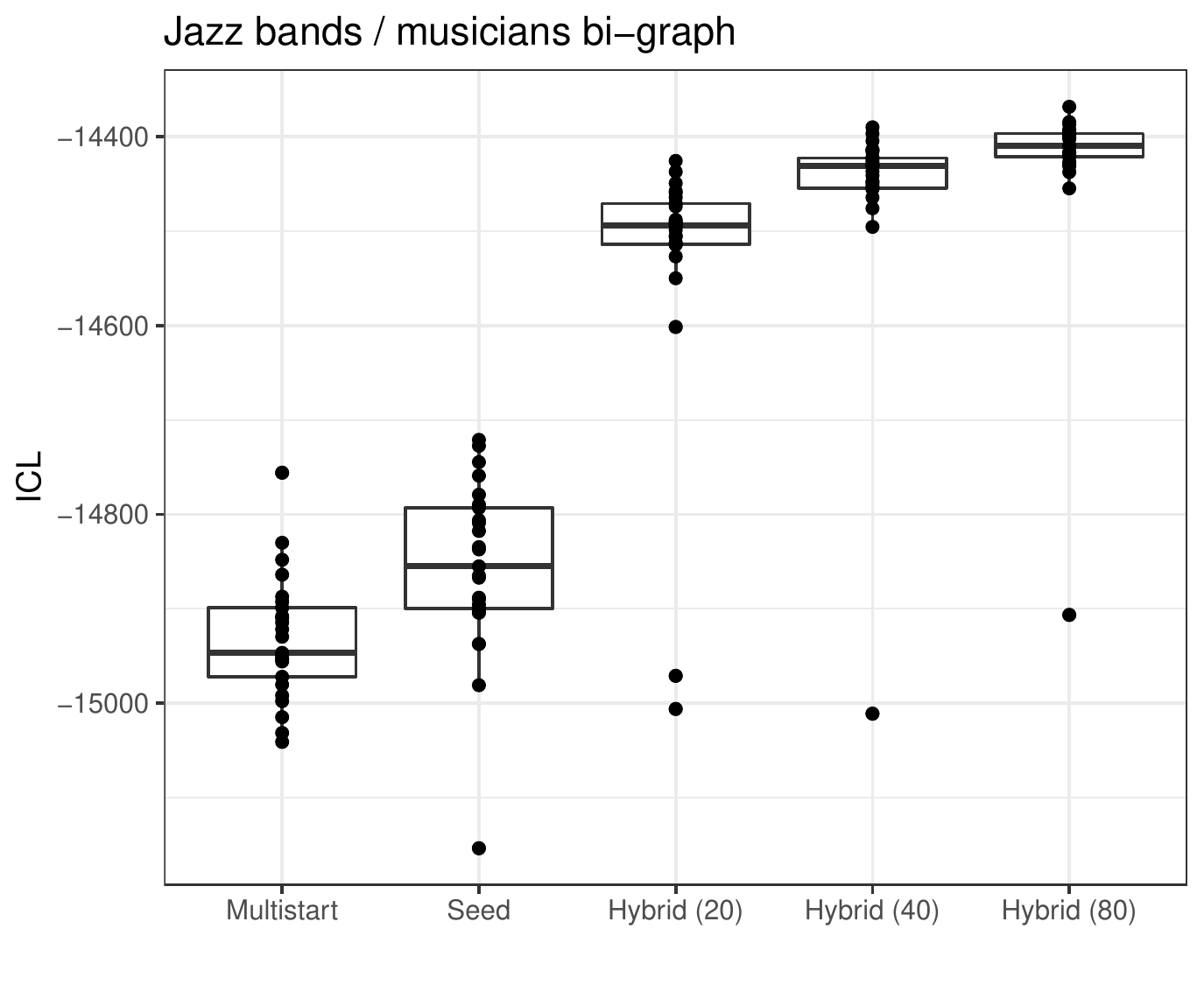}\\
		\end{tabular}
		\caption{Boxplots of the $\ICLex$ values obtained from 25 runs of the different algorithms on the six different datasets.}
		\label{fig:realdata_hybrid}
	\end{figure}
	
	The results are presented with boxplots in Figure~\ref{fig:realdata_hybrid} for all the methods that maximize the $\ICLex$. For all the datasets, the best results are achieved by the hybrid algorithm with a population of 80 partitions. For each experiment, while a bigger population size leads to better results with less variation, a small population size of 20 already achieves a significant improvement over the multiple and seeded strategies. Indeed, an important performance gap in terms of $\ICLex$ is visible between the three hybrid solutions and the two others. Moreover, some datasets like Jazz, Blogs and Political books highlights the interest of the multiple restart over the seeded strategy. This is expected for the experiments with directed networks (\textbf{Blogs}, \textbf{Books}), where the seed partitions are found using an undirected network model. Thus, it advocates for the use of directed model whenever possible for these datasets. This last experiment on the proposed hybrid algorithm clearly shows a benefit of using such an approach on real data. In the next section, we illustrate the interest of the hierarchical algorithm, giving a more detailed discussion about the clustering results on real datasets.

	\subsection{Hierarchical analysis of real datasets}
	\label{sub:ExpeHierarchical}
	In continuity with the motivating example of Figure \ref{fig:motivdendo}, the interest of the hierarchical procedure is illustrated on the real datasets introduced previously. Starting from the best solution of Algorithm \ref{alg:hybrid1}, with a population size of $40$, we build the hierarchy and the dendrogram for each of the examples. We start by describing the results on the four graph clustering datasets, then detailing the French parliament votes co-clustering one. 
	
	\paragraph{Newtork clustering} Figure \ref{fig:real-graph} shows the results of the proposed two-step methodology with the dc-SBM as the underlying model, highlighting its analytical and visual interest. Columns represent datasets and the first row corresponds to the adjacency matrices of each network, with the rows/columns arranged per cluster numbers and the color indicating the link density between clusters. Notice that clusters are reordered according to the leaf ordering of the dendrogram, bringing \textit{linked} clusters next to each other, enhancing the visualization of the block clustering. Next, the second row represents the cluster node link diagram, another representation of a graph clustering where the size of nodes is proportional to cluster size and the width of arrows to link density between clusters. Once again, we use the leaf ordering provided by the binary tree. The latter is then plotted as a dendrogram in the third row, emphasizing the amount of regularization (drop in $\alpha$) needed for each fusion.
	\begin{figure}[!ht]
		\centering
		\includegraphics[width=\textwidth]{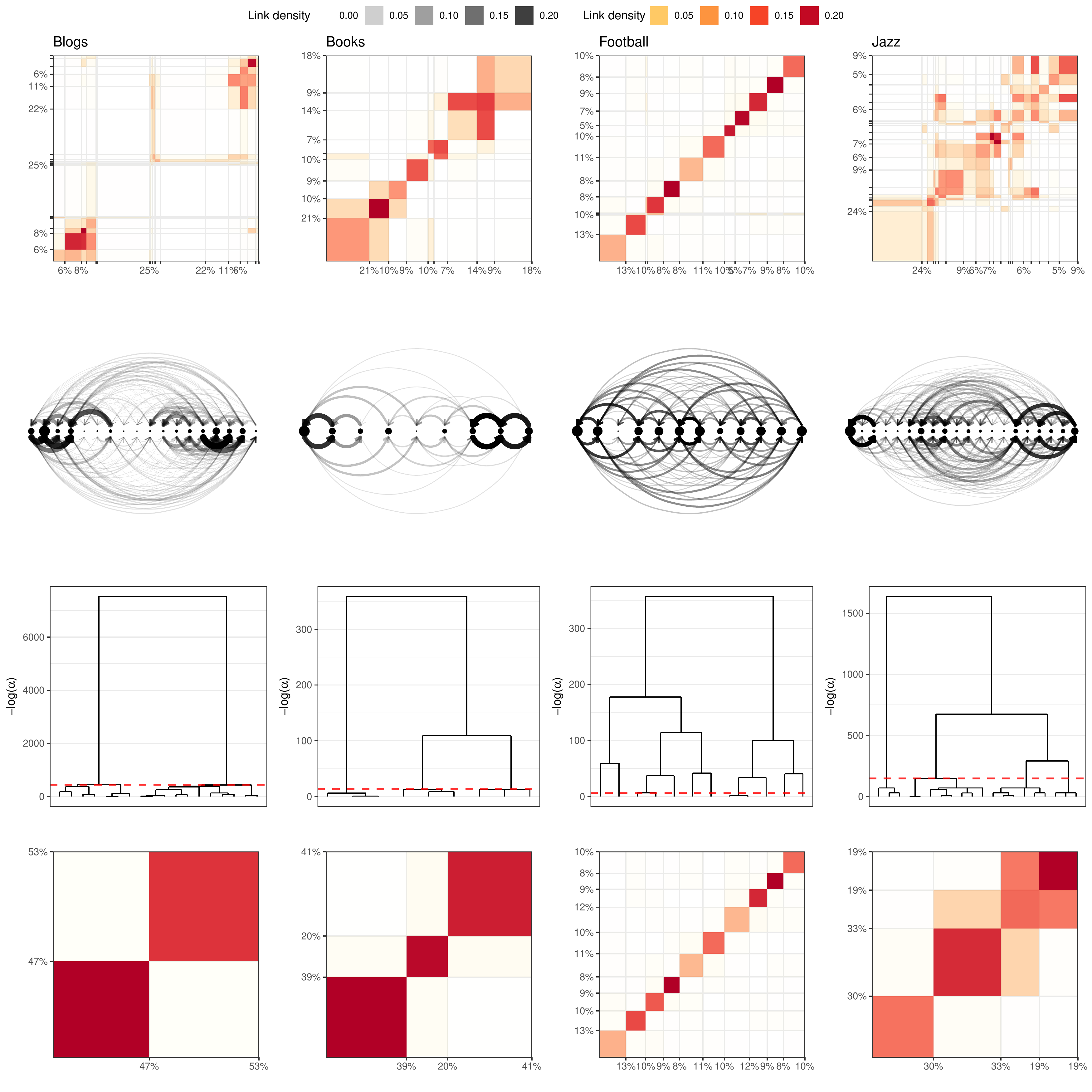}
		\caption{Illustration of the hierarchical agglomerative strategy on four real networks: blogs, books, football and jazz. First row: aggregated adjacency matrix according to the initial partition $\Clust^{(K)}$, with cluster reordering given by the leaf ordering of the dendrogram. Second row: node link diagram of $\Clust^{(K)}$. Third row: dendrogram of the hierarchy extracted from the initial partition. Fourth row: exploration of some clustering $\Clust^{(f_h)}$ alongside the hierarchy.}
		\label{fig:real-graph}
	\end{figure}
	
	In order to spot interesting levels in the dendrogram, we use a heuristic consisting in pruning the tree at a certain level $\alpha^{(f_h)}$ where the amounts of regularization needed for the next fusion is considered too important, relatively to the amount needed for past fusions. The fourth row of \Cref{fig:real-graph} represents the same adjacency matrices as in the first row, except the new clustering $\Clust_{f_h}$ is now used. For the Blogs network, starting from a solution with 18 clusters, the heuristic finds a lot of fusions for reasonable $\alpha$ levels, leaving 2 clusters at the selected level. The partition obtained at this level strongly aligns with the expected structure of this dataset: a divide between liberal/conservative blogs with an assortative structure of the two communities. Here, the NMI between the extracted partition and the labels provided with the dataset is 0.73. The method therefore allows to extract and analyse fine details with the initial partition and the bigger structures at another level of the dendrogram.

	 Likewise, for the Books dataset, the heuristic selects 3 clusters with assortative behavior which is the expected  structure for this network with liberal, conservative and a few neutral books. The corresponding partition has an NMI of 0.57 with the manual labelling of the books, and displays small differences in the "neutral" cluster.
	 
	 The Football network has a more pronounced and balanced community structure, with the initial partition $\Clust^{(11)}$ already strongly aligned with the additional information available on the football teams, i.e. their conference structure, which corresponds to an NMI of 0.86. There are a few independent teams that do not belong to any conference, which explains the observed differences. Remarkably, the dendrogram structure of this network is more balanced without a clear jump in $\log(\alpha)$. Still, the heuristic cuts the dendrogram after the first fusion at 10 clusters and the obtained partition is also close to the expected conference structure between the teams. 
	 
	 As for the Jazz network, it starts with 21 clusters and we propose to cut at 4 clusters according to the heuristic, with the corresponding $\Clust^{(4)}$ presenting an interesting block structure of three assortative communities and a fourth who bridges two of them. While this dataset does not possess side information to analyse the extracted structure quantitatively, the reordering of the $21$ initial clusters along with the dendrogram allows for a multi-level analysis from $\Clust^{(21)}$ to $\Clust^{(4)}$. This gives a complementary and clear view of the structure found in this graph with finer details available at the first levels of the hierarchy. 
	 
	 Overall, this highlights the relevance of the proposed hierarchical agglomeration in term of clustering quality and interpretability as well.

	\paragraph{Co-clustering on French assembly votes} We illustrate the hierarchical heuristic on the French assembly votes co-clustering dataset. The initial partition $\Clust^{(K)}$ found by Algorithm~\ref{alg:hybrid1} has 116 clusters divided in $70$ row clusters and $46$ columns clusters. These are quite large numbers for a dataset of this size, and one might want to explore solutions with fewer row clusters. As explained above, the hierarchical algorithm can build two separate dendrograms for rows and columns, which are linked by their merging sequence $(\alpha_f)_f$. Then, using the same heuristic on this regularization sequence, we chose to cut both dendrogram at the same level, thus determining a number of row and column clusters. In this example, the chosen level leaves the same number of $13$ rows and $13$ columns clusters. Inspecting the row clustering, we found it consistent with the true labels, which are the political party memberships. Some members of Parliament (MPs) in different opposition groups from the left (communists, socialists) are gathered in a single cluster, whereas MPs from the majority group (LREM) are split into $5$ different clusters, with some having centrists or right-wing opposition members. This agrees with the current separations and relationships in the French Parliament and the French political field. 
	
	\begin{figure}[!ht]
		\centering
		\includegraphics[width=\textwidth]{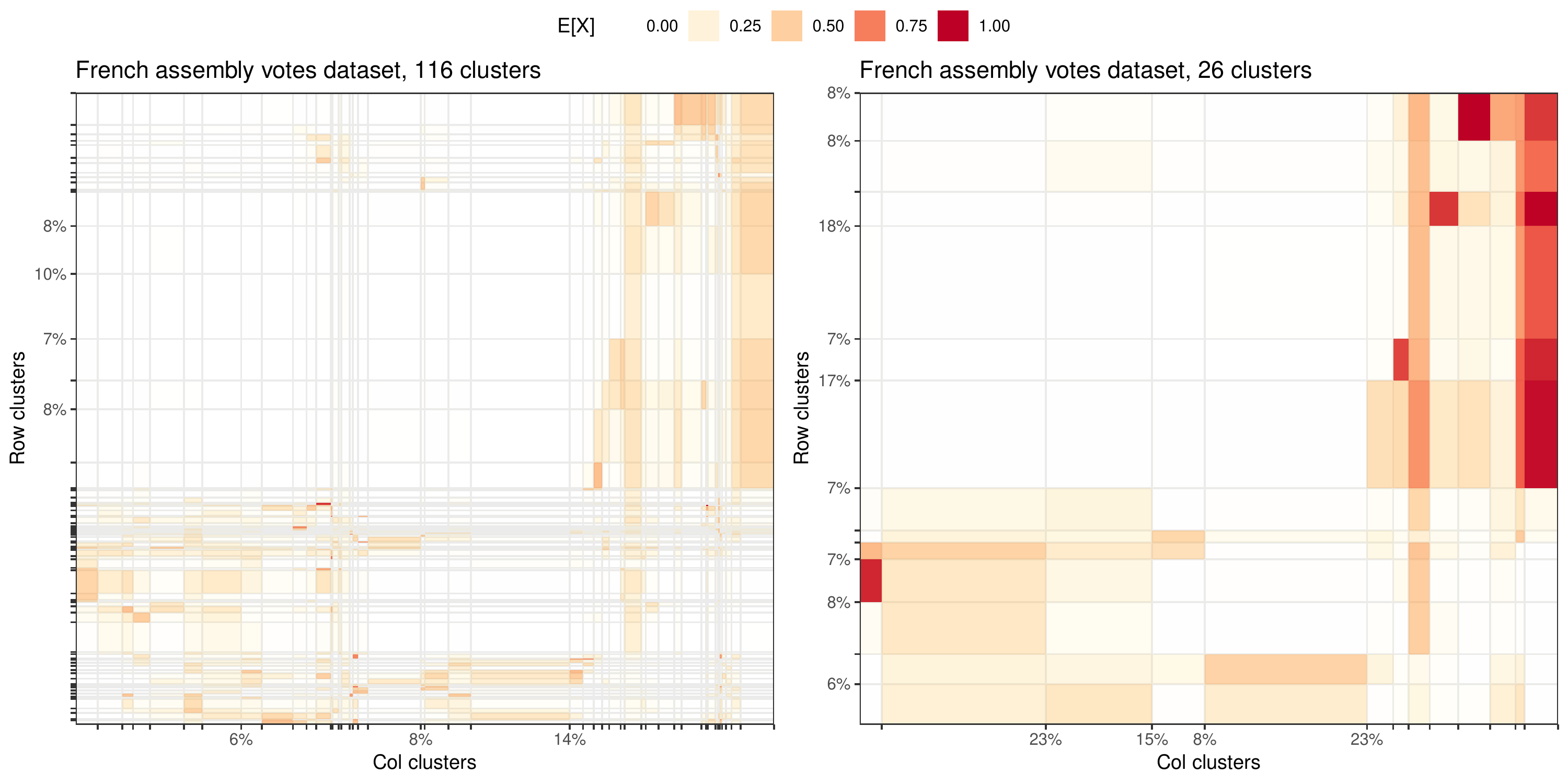}
		\caption{Block matrix representation of the \textbf{French Parliament} dataset after cluster reordering (left) and coarser clustering extraction (right).}
		\label{fig:real-pol-mat}
	\end{figure}
	
	\begin{figure}[!ht]
		\centering
		\includegraphics[width=\textwidth]{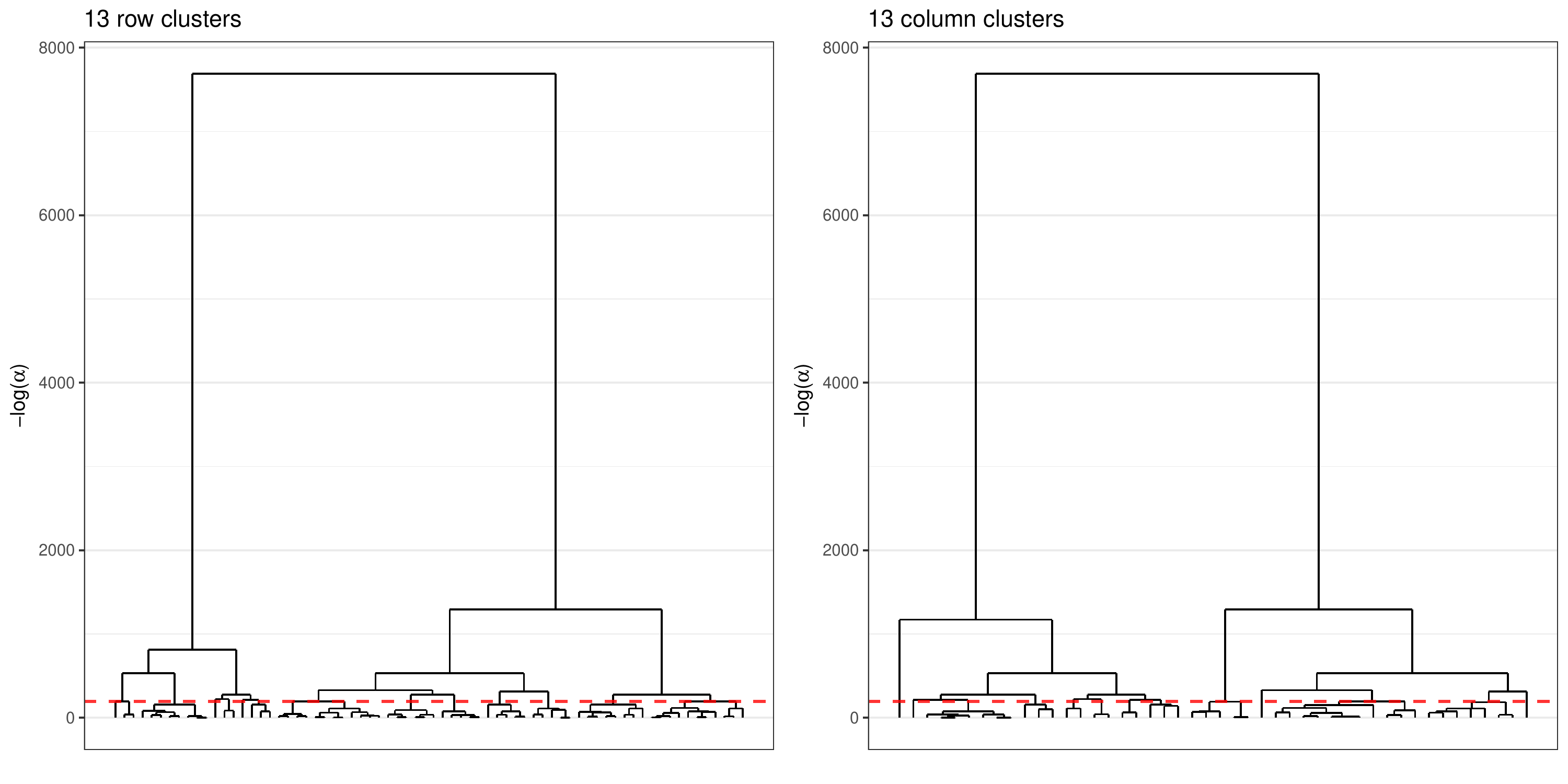}
		\caption{Row clusters dendrogram (left), and columns clusters dendrogram (right) for the \textbf{French Parliament} dataset. The dashed red line represents the height used to cut the dendrogram and to extract a coarser clustering.}
		\label{fig:real-pol-dendo}
	\end{figure}

\section{\revision{Related works}}
\revision{This paper introduced a two-fold contribution relying on two distinct optimization strategies, for which we now highlight connections with the existing literature.}

\paragraph{Local maxima and genetic clustering algorithms}
Apart from the work of \textcite{tessier2006evolutionary} for the latent class model, evolutionary algorithms were proposed for Gaussian model-based clustering, maximizing the (non integrated) classification likelihood \parencite{andrews2013using}, and in the context of feature selection \parencite{scrucca2016genetic}. More generally, the specific use of GAs for clustering problems is not new \parencite{Cole1998}, and we refer to \textcite{Hruschka2009} for a recent and detailed review on the subject.

\paragraph{Hierarchical clustering using the ICL}
Model-based hierarchical clustering extends the idea of non-parametric and similarity-based hierarchical clustering strategy, such as Ward's methods \parencite{ward1963hierarchical} or complete-link \parencite{sokal1958statistical} and single-link \parencite{sneath1957application} clustering. The first work of \textcite{murtagh1984fitting} extends Ward's criterion as the likelihood in an isotropic Gaussian mixture models, and was later extended to the general case of spectral constraints $\bS_k = \lambda_k \bD_k \bDelta_k \bD_k^\top$ \parencite{banfield1993model,fraley1998algorithms}. In this spirit, \textcite{zhong2003unified} proposed an extension of Ward's distance as the difference of log-likelihoods before and after a merge, along with ways to approximate it when the inference step is too costly to be done for each fusion.

More recently, model selection criteria were proposed as objective functions in hierarchical clustering algorithms. \textcite{heller2005bayesian} proposed a hierarchical Bayesian clustering algorithm, based on hypothesis testing. Marginal likelihoods of clusters are computed at each stage, using conjugate priors involving similar expressions as in the $\ICLex$. Explicitly working with a $\ICLbic$ criterion, \textcite{Baudry2010} proposed a soft hierarchical clustering algorithm for finite mixture models. Relying on an asymptotic approximation rather than exact derivation, it chooses the merge inducing lowest posterior entropy for the cluster memberships probabilities. Thus, the latter is used to assess clustering quality, and the output is a hierarchy of soft partitions. In the context of network analysis, \textcite{peixoto2014hierarchical} proposed a greedy hierarchical clustering algorithm for a hierarchical formulation of the SBM, using another model selection criterion: the \textit{description length}. Although the criterion differs, the author shows that it matches the $\ICLex$ when the prior on the connection probabilities of the SBM is replaced by a nested sequence of priors and hyper-priors.
	
	\section{Conclusion}
	
	In this paper, we proposed a new methodology for model-based hierarchical clustering with discrete latent variables models, based on two related contributions. The first one uses a hybrid genetic algorithm to jointly cluster the data and select the number of cluster $\K$. The second one uses the former as an initialization and completes the hierarchy  by including a Dirichlet hyper-parameter $\alpha$ in the objective criterion, allowing to access coarser partitions. Both methods share the ICL as an objective criterion to maximize, and their interest lies on their computational efficiency as well as the wide variety of models they can handle. Numerical experiments assessed the interest and superiority of the genetic clustering algorithm over existing methods, for some of the most common models for discrete data or graphs clustering. In addition, experiments on real datasets were conducted to illustrate the interest of the hierarchical algorithm in real-world applications for both clustering and co-clustering. The resulting hierarchy may be visualized as a dendrogram, and explored as well as the amount of regularization needed for each fusion. Moreover, we illustrated how the leaf ordering of the dendrogram may be used to reorder clusters in the initial partition, enhancing the visualization of any clustering. 

	Regarding further works, we plan to focus on the special case of Gaussian mixtures. Indeed, the latter perfectly fit into the DLVM framework and an exact ICL is available \parencite{Bertoletti2015}. However, the difficulty of setting uninformative priors must be addressed carefully, as the clustering results is greatly influenced by these.
	
	\section*{Acknowledgement}
	The authors would like to thank the editor and the two anonymous referees for their fruitful comments which helped to improve this paper. 
	
	\clearpage
	
	\appendix

	\section{Deriving exact ICL: application to some DLVMs}
	\label{sec:models}

     This appendix discusses the detail of $\ICLex$ derivation for the discrete latent variable models used in the experiments of \Cref{sec:exp}. First, we detail how the marginal distribution of $\Clust$ is obtained in \Cref{eq:ICLexdec}: this part is common to all DLVMs. Then, the only quantity needed to explicit a particular model is $\log \p(\Obs|\Clust, \bBeta)$, namely the supposed generative model at hand in \Cref{HCICL:eq:DLVM}. The latter additionally depends on the prior distribution on $\param$, which is governed by the $\bBeta$ hyper-parameters, and we discuss model dependent specifications of the latter. 
	
	\subsection{Marginal distribution of \texorpdfstring{$\Clust$}{Z}: Dirichlet-Multinomial conjugacy}
	\label{appendix:HCICL:p(Z)}
	We recall the expression of $\ICLex$ in \Cref{eq:ICLex}:
	\begin{equation*}
		\ICLex(\Clust) = \log \p (\Obs \mid \Clust, \bBeta) + \log \p (\Clust \mid \balpha).
	\end{equation*}
	As explained in the introduction, the second term is analytically tractable when $\bPi \sim \Dir_{\K}(\balpha = (\alpha, \ldots, \alpha))$, leading to the expression in \Cref{eq:ICLexdec}. This is obtained by an application of standard Dirichlet-Multinomial conjugacy, which is detailed in the following for the sake of completeness. 

	Let us denote by $C(\bm{t})$ the normalization constant of the Dirichlet distribution:
	\[
	C(\bm{t}) = \frac{\prod_{k=1}^\K \Gamma(t_k)}{\Gamma(\sum_{k=1}^{\K} t_k)}.
	\]
	Then, we want to compute the following integral:
	\begin{align*}
		\p(\Clust \mid \alpha) &= \int_{\bPi} \p(\Clust \mid \bPi) \p(\bPi \mid \alpha) \dif \bPi, \\
		&= \int_{\bPi} \left(\prod_{i=1}^{\nb} \Mult_\K(\clust_{i} \mid 1, \bPi) \right) \Dir_\K(\bPi \mid \balpha) \dif \bPi, \\
		&= \int_{\bPi} \left( \prod_{k=1}^{\K} \pi_k^{\sum_i \rawclust_{ik}} \right)  \frac{1}{C(\balpha)} \prod_{k=1}^{\K} \pi_k^{\alpha - 1}  \dif \bPi, \\
		&= \frac{C(\bm{\nb + \balpha })}{C(\balpha)} \int_{\bPi} \Dir_\K(\bPi \mid \balpha + \bm{\nb}) \dif \bPi, \\
		& = \frac{C(\bm{\nb + \balpha})}{C(\balpha)},
	\end{align*}
	with $\nb_k = \sum_i \rawclust_{ik}$. Thus, we obtain the desired result as:
	\begin{equation}
		\log \p(\Clust \mid \alpha) = \log \left\{ \frac{\Gamma(\alpha \K) \textstyle \prod_k \Gamma(\alpha + n_k )}{\Gamma(\alpha)^{\K} \Gamma(\nb + \alpha \K)} \right\}	
	\end{equation}
	
	\subsection{Mixture of multinomials}
	
	Multivariate count data arise in many scientific fields in the form of frequency counts, such as word occurrence in text analysis, read counts in RNA-seq data, or species abundance data in ecology. Formally, an observation $\obs_{i}$ is supposed to be a count vector in $\mathbb{N}^\dim$, where $\rawobs_{ij}$ represents the count of modality $j$, with total count $\totalcount_i = \sum_{j=1}^{\dim} \rawobs_{ij}$. 
	Here, we consider the mixture of multinomials (MoM) model, which is a mixture model for the clustering of discrete data. In a Bayesian context, we define a symmetric conjugate Dirichlet prior on each parameter $\param_k$ and the generative model of Equation~\eqref{eq:MixtureModels} is given by:
	\begin{align}
	\label{HCICL:eq:MoM}
	\param_{k}&\sim \mathcal{D}_{\dim}(\bBeta = (\beta, \ldots, \beta)),\nonumber\\
	\obs_{i}|\rawclust_{ik}=1,\param &\sim\mathcal{M}_{\dim}(\totalcount_i, \param_{k}).
	\end{align}
	Then, each parameter $\param_k$ can be marginalized out exactly, giving a Dirichlet-multinomial distribution \parencite{minka2000estimating} per cluster.
	\begin{proposition}
		\label{prop:MoMICL} Under the mixture of multinomials model of \Cref{HCICL:eq:MoM}, we have:
		\begin{equation}
		\label{eq:MoMICL}
		\log \p(\Obs|\Clust)= \sum_k\log\left( \dfrac{\Gamma( \beta \dim)\prod\limits_{j=1}^\dim \Gamma( o_{kj}+\beta )}{\Gamma(\beta)^\dim \, \Gamma(\totalcount_k + \beta \dim)}\right) + \log B(\Obs),
		\end{equation}
		with $o_{kj}=\sum_{i=1}^{\nb}\rawclust_{ik}\rawobs_{ij}$, $\totalcount_k=\sum_{j=1}^{\dim} o_{kj}$ and $B(\Obs)$ is a constant that does not depend on $\Clust$ or $\bBeta$.
	\end{proposition}
		\begin{proof}[Proof of \Cref{prop:MoMICL}]
		Here, $\param = (\param_k)_k \in \Simplex_\dim^\K$ and the conditional likelihood, given $\Clust$, of the MoM generative model is:
		\begin{align*}
			\p(\Obs \mid \Clust, \param) & = \prod_{k=1}^{\K} \prod_{i=1}^{\nb} \Mult_{\dim}(\obs_i \mid \totalcount_{i}, \param_k)^{\rawclust_{ik}}, 
		\end{align*}
		We wish to integrate out the parameters $\param \sim \otimes_k \Dir_\dim(\bBeta = (\beta, \ldots, \beta))$. A use of Fubini's formula allows leveraging Dirichlet-Multinomial conjugacy for $\K$ different integrals:
		\begin{align*}
			\p(\Obs \mid \Clust, \beta) &= \int_\param \p(\Obs \mid \Clust, \param) \p(\param \mid \bBeta) \dif \param , \\
			& = \prod_{k=1}^{\K} \int_{\param_k}  \prod_{i=1}^{\nb} \Mult_{\dim}(\obs_i \mid \totalcount_{i}, \param_k)^{\rawclust_{ik}} \Dir_{\dim}(\param_k \mid \bBeta) \dif \param_{k}, \\
			&=  \frac{1}{\prod_{i,j} \rawobs_{ij}!} \prod_{k=1}^{\K}   \int_{\param_k} \left( \prod_{j=1}^{\dim} \rawparam_{kj}^{\sum_i \rawclust_{ik} \rawobs_{ij}} \right) \frac{1}{C(\bBeta)} \prod_{j=1}^{\dim} \rawparam_{kj}^{\beta - 1} \dif \param_k, \\
			&=  \frac{1}{\prod_{i,j} \rawobs_{ij}!} \prod_{k=1}^{\K}  \frac{C(\bm{o}_k)}{C(\bBeta)} \int_{\param_k} \Dir_{\dim}(\param_k \mid \bBeta + \bm{o}_k) \dif \param_k , \\
			& = \frac{1}{\prod_{i,j} \rawobs_{ij}!} \times \prod_{k=1}^{\K}  \dfrac{\Gamma( \beta \dim)\prod\limits_{j=1}^\dim \Gamma( o_{kj}+\beta )}{\Gamma(\beta)^\dim \, \Gamma(\totalcount_k + \beta \dim)},
		\end{align*}
		with $o_{kj} = \sum_{i=1}^{\nb} \rawclust_{ik}\rawobs_{ij}$ and $\totalcount_k = \sum_{j=1}^{\dim} o_{kj}$. Finally, denote 
			\[
			B(\Obs) = \frac{1}{\prod_{i,j} \rawobs_{ij}!},
			\]
			which solely depends on $\Obs$. Thus, taking the log concludes the proof.
	\end{proof}
	
	\textcite{tessier2006evolutionary, Biernacki2010} analogously derived an exact ICL criterion  for the latent class model (LCM) which is closely related to the MoM model, and we emphasize that the LCM model also fits in the proposed framework. The derivation of greedy updates for merge or swap moves does not present difficulties for these models. As for setting the $\beta$ hyper-parameter, uninformative prior or Jeffreys prior can be used by setting $\beta$ to $1$ or $\frac{1}{2}$.

	\subsection{Stochastic block models and degree correction}
	
We now describe the derivations for the standard binary SBM described in \Cref{eq:SBM} as well as its degree-corrected variant.

	\paragraph{Binary SBM}
	In the binary SBM framework, $\rawobs_{ij}$ are Bernoulli random variables indicating the presence or absence of an edge. As mentioned above, the probability of a connection between the nodes $i$ and $j$ only depends on their cluster assignments $\clust_{i}$ and $\clust_{j}$. Hence, there is a connection probability parameter $\rawparam_{kl}$ for each pair of clusters. Ultimately, a Bayesian formulation of SBM is given by: 
	\begin{align}
	\rawparam_{kl}&\sim \Betad(\eta^0,\zeta^0),\nonumber\\
	\rawobs_{ij}|\rawclust_{ik}\rawclust_{jl}=1,\param&\sim \Bernoulli(\rawparam_{kl}),
	\end{align}
	where the Beta prior on the connection probabilities is used as a conjugate of the Bernoulli distribution with hyper-parameter $\bBeta = (\eta^0, \zeta^0)$. \textcite{Come2015} derived an exact ICL criterion for this model, relying on Beta-Bernoulli conjugacy.
	\begin{proposition}[Proof in \textcite{Come2015}, Appendix A]
		\label{prop:SBMICL}
		Under the SBM model, we have:
		\begin{equation}
		\label{eq:SBMICL}
		\log \p(\Obs|\Clust)=\sum_{k,l}\log\left(\frac{\Gamma(\eta^0+\zeta^0)\Gamma(\eta_{kl})\Gamma(\zeta_{kl})}{\Gamma(\eta^0)\Gamma(\zeta^0)\Gamma(\eta_{kl}+\zeta_{kl})}\right),\\
		\end{equation}
		with $\eta_{kl}= \eta^0 + \sum_{i \neq j}\rawclust_{ik}\rawclust_{jl}\rawobs_{ij}$ and $\zeta_{kl}= \zeta^0 + \sum_{i \neq j}\rawclust_{ik}\rawclust_{jl}(1-\rawobs_{ij})$.
	\end{proposition}
	Again, a commonly accepted value for setting the hyper-parameter $\bBeta$ is $\eta^0 = \zeta^0 = 1$ or $1/2$, for a uniform or Jeffreys prior respectively.
	
	\paragraph{Degree correction}
	Real world networks tend to exhibit a specific degree distribution, with some nodes having a number of links greatly superior to the average. In the SBM, all nodes inside a cluster are statistically equivalent, hence a simple SBM model may have some difficulty in reproducing such heterogeneous degree distributions. 
	\textcite{karrer2011stochastic} proposed a slight modification of the SBM to respect the degree sequences of the observed graph. It can be expressed as an SBM generative model, where the connection probability between two nodes now also depends on node parameters $\bm{\Phi}$ in order to introduce disparity between the nodes. This new model is called degree-corrected stochastic block model (dc-SBM). We introduce a slightly more general version of this model for directed graphs similar to the model introduced in~\textcite{zhu2014}, where the parameters $\bm{\Phi}^-$ and $\bm{\Phi}^+$ govern the out-degree and in-degree distributions of nodes respectively. Then,  defining the degree prior distributions as in \textcite{Newman2016} and \textcite{Riolo2017}, the model writes as follows:
	\begin{align}
	\label{model:dcSBM}
	\bm{\Omega}_{kl}&\sim \Exponential(\beta^{-1}),\nonumber\\
	\bm{\Phi}^+_k, \, \bm{\Phi}^-_k  \mid \Clust  &\sim \Uniform(\mathbb{S}_k),  \\
	\rawobs_{ij}|\rawclust_{ik}\rawclust_{jl}=1,\bm{\Omega},\bm{\Phi}&\sim \Poisson(\Phi^-_i\Omega_{kl} \Phi^+_j ). \nonumber
	\end{align}
	Here, $\bm{\Phi}_k^{\cdot} = (\Phi_i^{\cdot})_{i:\rawclust_{ik}=1}$, and $\mathbb{S}_k = n_k \Delta_{n_k}$ the rescaled simplex of dimension $(n_k-1)$ induced by the constraints $\sum_i \Phi_i^{\cdot} \rawclust_{ik} = n_k$. The latter must be set for the model to be identifiable. In this model the Bernoulli distribution of edges is replaced by a Poisson, in part to ease the computations, and the exponential distribution is used to leverage standard Gamma-Poisson conjugacy as $\Exponential(\beta^{-1}) = \Gamma(1, \beta)$. Thus, the $\ICLex$ of this model can be derived as detailed in the following proposition.
	\begin{proposition}
		\label{prop:DCSBMICL}
		Under the dc-SBM model we have:
		\begin{align}
		\label{eq:DCSBMICL}
		\log \p(\Obs|\Clust)= & \sum_k\log\left(\frac{(n_k-1)!\,n_{k}^{dg_k^+}}{(n_k+dg_k^+-1)!}\frac{(n_k-1)!\,n_{k}^{dg_k^-}}{(n_k+dg_k^--1)!}\right)\nonumber\\
		& + \sum_{k,l}\log\left(\frac{(\nu_{kl})! \, \beta^{\nu_{kl}}}{ \left(\beta n_k n_l + 1\right)^{\nu_{kl} + 1}}\right)+ \log B(\Obs),
		\end{align}
		where $\nu_{kl} = \sum_{i,j} z_{ik}z_{jl}\rawobs_{ij}$ is the total counts in block $(k,l)$, $d_i^-=\sum_j \rawobs_{ij}$ and $d_j^+=\sum_i\rawobs_{ij}$ correspond to node $i$ out-degree and in-degree respectively, and $dg_k^-$, $dg_k^+$ to their sums in cluster $k$. $B(\Obs)$ is a constant detailed in the Appendix, that does not depend on $\Clust$ or $\bBeta$.
	\end{proposition}
	\begin{proof}[Proof of \cref{prop:DCSBMICL}]
	Putting $\param = (\bm{\Phi}^+, \bm{\Phi}^-, \bm{\Omega})$, the conditional likelihood, given $(\Clust, \param)$, of the generative model described in Equation~\eqref{model:dcSBM} writes as:
	\begin{align}
	\label{eq:conditionalllhood}
	\p(\Obs \mid \Clust, \param) &=  \prod_{i,j}^N \prod_{k,l}^K \Poisson(\rawobs_{ij} \mid \Phi_i^+ \Phi_j^- \Omega_{kl})^{z_{ik}z_{jl}},   \nonumber \\ 
	&=   \frac{1}{\prod_{i,j} \rawobs_{ij}!} \times \prod_{i,j}^N (\Phi_i^+ \Phi_j^- )^{\rawobs_{ij}} \prod_{k,l}^K  \Omega_{kl}^{z_{ik}z_{jl}\rawobs_{ij}}  \exp( \Phi_i^+ \Phi_j^- \Omega_{kl} z_{ik}z_{jl}), \nonumber \\
	&=  \frac{1}{\prod_{i,j} \rawobs_{ij}!} \times \prod_i (\Phi^+_i)^{d_i^+} \times \prod_i (\Phi^-_i)^{d_i^-} \times \prod_{k,l}\Omega_{kl}^{\nu_{kl}}\exp(-n_kn_l\Omega_{kl}),
	\end{align}
	 $d_i^-=\sum_j \rawobs_{ij}$, $d_j^+=\sum_i\rawobs_{ij}$ and $\nu_{kl} = \sum_{i,j} z_{ik}z_{jl}\rawobs_{ij} $. Calculating the $\ICLex$ implies to integrate over $\param$. Notice that Equation \eqref{eq:conditionalllhood} is separable as the product of two parts, one depending on $\bm{\Phi}$ and the other on $\bm{\Omega}$. In the following, we detail calculations separately for both parts. 
	\paragraph{Integrating over $\bm{\Phi}$} Recall that  $\bm{\Phi}_k^{\cdot} = (\phi_i^{\cdot})_{i:z_{ik}=1}$ and that:
	\begin{equation}
	\p(\bm{\Phi}_k^{\cdot}) = \frac{1}{\textrm{vol}(\mathcal{S}_k)} \mathds{1}_{\mathcal{S}_k}(\bm{\Phi}_k^{\cdot}) , \quad \textrm{with } \mathbb{S}_k = \left\{\bm{\Phi}_k^{\cdot} \in\mathbb{R}^{n_k}:\sum_{i:z_{ik}=1}\frac{\Phi_i^{\cdot}}{n_k}=1\right\},
	\end{equation}
	which is simply the simplex of dimension $(n_k-1)$, rescaled by a factor $n_k$. Hence the volume of $\mathbb{S}_k$ is given by:
	\begin{equation}
	\int_{\mathbb{S}_k} \dif \bm{\Phi}_k^{\cdot} =\frac{n_k^{n_k}}{(n_k-1)!}.
	\end{equation}
	The situation is symmetric for $\bm{\Phi}_k^+$ or $\bm{\Phi}_k^-$. Thus, calculations are detailed only for the former. One needs to compute: 
	\begin{align}
	\frac{(n_k-1)!}{n_k^{n_k}}\int_{\mathbb{S}_k}\prod_{i:z_{ik}=1}\left(\Phi^+_i\right)^{d_i^+}\dif\bm{\Phi}_k^+
	&=\frac{(n_k-1)!}{n_k^{n_k}}n_{k}^{d_k^+}\int_{\mathbb{S}_k}\prod_{i:z_{ik}=1}\left(\frac{\Phi^+_i}{n_k}\right)^{d_i^+}\dif\mathbf{\bm{\Phi}}_k^+, \nonumber \\
	&=\frac{(n_k-1)!}{n_k^{n_k}}n_{k}^{d_k^+}\int_{\Delta_{n_k}} \prod_{i:z_{ik}=1}\left(x_{i}\right)^{d_i^+}|n_k\mathbf{I}_{n_k}| \dif \mathbf{x} , \nonumber \\
	&=\frac{n_k^{n_k}}{n_k^{n_k}}(n_k-1)! \, n_{k}^{d_k^+} C(\bm{a}_k) \int_{\Delta_{n_k}} \Dir_{n_k}\left(\scores \mid \bm{a}_k = (d_i + 1)_{i:\rawclust_{ik} =1} \right) \dif \mathbf{x} , \nonumber\\
	&=(n_k-1)! \, n_{k}^{d_k^+}\frac{\prod_{i:z_{ik}=1}d_i^+!}{(n_k+d_k^+-1)!} , \nonumber\\
	&=\frac{(n_k-1)!}{(n_k+d_k^+-1)!}n_{k}^{d_k^+}\prod_{i:z_{ik}=1}d_i^+! ,
	\end{align}
	with $d_k^+ = \sum_{i} \rawclust_{ik} d_i^+$. Then,
	\begin{align}
	\int_{\Phi^+} \p(\Phi^+) \, \prod_i \left(\Phi^+_i\right)^{d_i^+} \dif \bm{\Phi}^+ &= \int_{\prod_k \mathbb{S}_k}  \prod_k \frac{(n_k-1)!}{n_k^{n_k}} \prod_{i:z_{ik}=1} \left(\Phi^+_i\right)^{d_i^+} \dif \, (\bm{\Phi}_1^+, \ldots, \bm{\Phi}_K^+), \nonumber \\
	& =  \prod_k \frac{(n_k-1)!}{n_k^{n_k}} \int_{\mathbb{S}_k}   \prod_{i:z_{ik}=1} \left(\Phi^+_i\right)^{d_i^+} \dif \bm{\Phi}_k^+ \nonumber, \\
	& =  \prod_k\frac{(n_k-1)!}{(n_k+d_k^+-1)!}n_{k}^{d_k^+}\prod_{i}d_i^+! .
	\end{align}
	\paragraph{Integrating over $\bm{\Omega}$} This is done using a standard Gamma-Poisson conjugacy in each pair of clusters. Indeed:
	\begin{align}
	\int_{\Omega_{kl}} \p(\Omega_{kl}) \Omega_{kl}^{\nu_{kl}}\exp(-n_kn_l\Omega_{kl}) \dif \Omega_{kl} &= \int_{\Omega_{kl}} \frac{1}{\beta} \Omega_{kl}^{\nu_{kl}} \exp\left( - (n_kn_l + \frac{1}{\beta})\Omega_{kl} \right) \dif \Omega_{kl} , \nonumber \\
	&= \frac{\Gamma(\nu_{kl} + 1 )}{\beta \left(n_k n_l + \beta^{-1}\right)^{\nu_{kl} + 1}} , \nonumber \\
	&= \frac{\nu_{kl}! \, \beta^{\nu_{kl}}}{ \left(\beta n_k n_l + 1\right)^{\nu_{kl} + 1}} .
	\end{align}
	Utimately, we have:
	\begin{align}
	\p(\Obs \mid \Clust) &= \int_{\param} \p(\Obs, \param \mid \Clust) \dif \param, \nonumber \\
	&= \frac{\prod_{i}d_i^+!d_i^-!}{\prod_{ij} \rawobs_{ij}!} \, \prod_k\frac{(n_k-1)!}{(n_k+d_k^+-1)!}n_{k}^{d_k^+}\frac{(n_k-1)!}{(n_k+d_k^--1)!}n_{k}^{d_k^-} \nonumber \\
	&  \qquad  \qquad  \qquad  \qquad \times \prod_{k,l} \frac{ \nu_{kl}! \, \beta^{\nu_{kl}}}{ \left(\beta n_k n_l + 1\right)^{\nu_{kl} + 1}}.
	\end{align}
	Putting 
	\[
	B(\Obs) = 	 \frac{\prod_{i}d_i^+!d_i^-!}{\prod_{ij} \rawobs_{ij}!},
	\]
	and noticing that the latter does not depend on the partition $\Clust$ concludes the proof.
\end{proof}
	
	Contrary to the previous models where proper Jeffreys or uniform prior could be used, the exponential distribution does not admit a conventional uninformative prior. An acceptable solution to fix $\beta$ is, however, proposed in \textcite{Newman2016}, where the authors use the mean connection probability of the network. From a practical point of view, deriving swap and merge updates is also quite easy for these models, even though some care is needed to avoid unnecessary computations \parencite[Appendices B and C]{Come2015} and can be done efficiently using sparse matrices.  
	
	\subsection{Co-clustering and latent block model} 
	\label{subsec:LBM}
	Co-clustering aims at clustering simultaneously the rows and columns of a data matrix $\Obs$ of size $ \nb \times \dim$ into homogeneous groups. For example, in text analysis one may be interested into grouping documents and words together. The latent block model \parencite[LBM,][]{govaert2010}, introduced in \Cref{eq:LBM}, is a popular generative model to perform such task, forming a flexible class of models depending on the supposed observational model \parencite{Wyse2017}. The main feature of the LBM is its block generation hypothesis:
	\begin{align*}
	&\clustrow_i \sim \Mult_{\Kr}(1, \bPirow), \; \clustcol_j \sim \Mult_{\Kc}(1, \bPicol), \\
	&\rawobs_{ij} \mid \rawclustrow_{ik} \rawclustcol_{jl}=1,\param \sim \p(\cdot \mid \param_{kl}).
	\end{align*}
	Here, $\Clustrow$ and $\Clustcol$ are binary matrices defining a partition of the $\nb$ rows in $\Kr$ clusters and of the $\dim$ columns into $\Kc$ clusters respectively. The LBM may be handled similarly as other DLVMs, with a slight variation of the prior to handle the bipartition aspect:
	\begin{equation}
		\p(\bPi \mid \alpha) = \Dir_{\Kr}(\bPirow \mid \alpha) \times \Dir_{\Kc}(\bPicol \mid \alpha).
	\end{equation}
	With such a prior, the likelihood of the bipartition integrated with respect to $\bPi$ is factorized $\p(\Clust \mid \alpha) = \p(\Clustrow \mid \alpha) \p(\Clustcol \mid \alpha)$ and writes as:
	\begin{equation}
	\label{eq:LBMprior}
	\p(\Clust \mid \alpha) = \dfrac{\Gamma(\alpha\Kr)\prod\limits_{k=1}^{\Kr} \Gamma(\alpha + n_k)}{\Gamma(\alpha)^{\Kr} \, \Gamma(\nb + \alpha\, \Kr)} \times \dfrac{\Gamma(\alpha\Kc)\prod\limits_{l=1}^{\Kc} \Gamma(\alpha + m_l)}{\Gamma(\alpha)^{\Kc} \, \Gamma(\dim + \alpha\, \Kc)}.
	\end{equation}
	Again, this part is common to any LBM, and independent on the observational model at hand. Thus, the only quantity needed to derive $\ICLex$ for the LBM is $\log \p (\Obs \mid \Clust, \bBeta)$. The latter is often explicit when working with standard distributions for $\rawobs_{ij}$, leveraging on known conjugacy results. This is notably the case for standard discrete data distributions using Beta-Bernoulli or Gamma-Poisson conjugacy. Other types of distributions may be considered, \textit{e.g.} for continuous data $\rawobs_{ij}$, and exponential family distributions are good candidates to derive natural conjugate priors on $\param_{kl}$.
	
	Moreover, as already emphasized, the SBM and LBM are very similar and a degree-corrected LBM can also be derived for discrete Poisson observations as follows:
	\begin{align}
	\label{model:dcLBM}
	\bm{\Omega}_{kl}&\sim \Exponential(\beta^{-1}),\nonumber\\
	\bm{\Phi}^r_k \mid \Clustrow  &\sim \Uniform(\mathbb{S}_k), \nonumber \\
	\bm{\Phi}^c_l  \mid \Clustcol  &\sim \Uniform(\mathbb{S}_l),  \\
	\rawobs_{ij}|\rawclustrow_{ik}\rawclustcol_{jl}=1,\bm{\Omega},\bm{\Phi}^r,\bm{\Phi}^c&\sim \Poisson(\Phi^r_i\Omega_{kl} \Phi^c_j ). \nonumber
	\end{align}
	Then, an $\ICLex$ can be derived which closely resembles the one of Proposition \ref{prop:DCSBMICL}, using similar arguments and calculations. 
	\begin{proposition}
		\label{prop:DCLBMICL}
		Under the dc-LBM model we have:
		\begin{align}
		\label{eq:DCLBMICL}
		\log \p(\Obs|\Clust) = &  \sum_k\log\left(\frac{(n_k-1)!\,n_{k}^{r_k}}{(n_k+r_k-1)!}\right) + \sum_l \log \left(\frac{(m_l-1)! \, m_{l}^{c_l}}{(m_l+c_l-1)!}\right) \nonumber \\
		& +  \sum_{k,l} \log\left(\frac{\nu_{kl}! }{ \left(\beta n_k m_l + 1\right)^{\nu_{kl} + 1}}\right) + \log B(\Obs),
		\end{align}
		where $\nu_{kl} = \sum_{i,j} \rawclustrow_{ik} \rawclustcol_{jl}\rawobs_{ij}$. Here,  $r_k=\sum_{ij} \rawclustrow_{ik} \rawobs_{ij}$ and $c_l=\sum_{ij} \rawclustcol_{jl} \rawobs_{ij}$ correspond to row and column cluster degrees, and $B(\Obs)$ is a constant detailed in the Appendix, that does not depend on $\Clust$ or $\bBeta$.
	\end{proposition}
	\begin{proof}[Proof of \Cref{prop:DCLBMICL}]
	Putting $\param = (\bm{\Omega},\bm{\Phi}^r,\bm{\Phi}^c)$, the conditional likelihood, given $\Clust$, of the generative model described in Equation~\eqref{model:dcLBM} writes as:
	\begin{align}
	\label{eq:dcLBMllhood}
	\log \p(\Obs \mid \Clust, \param) & = \prod_{i=1}^{\nb} \prod_{j=1}^{\dim} \prod_{k}^{\Kr} \prod_{l}^{\Kc} \Poisson(\rawobs_{ij} \mid \Phi_i^r \Phi_j^c \Omega_{kl})^{z_{ik}^c z_{jl}^r},   \nonumber \\ 
	\end{align}
	Calculations for each of the term in \Cref{eq:DCLBMICL} are similar to \Cref{subsec:LBM}, with a slight difference in the $B(\Obs)$ term. Indeed, the out (resp. in) degrees are now replaced by rows (resp. columns) degrees:
	\begin{equation}
	B(\Obs) = \frac{\prod_ir_i!\prod_jc_j!}{\prod_{i,j}\rawobs_{ij}!},
	\end{equation}
	with $r_i=\sum_j\rawobs_{ij}$ and $c_j=\sum_i\rawobs_{ij}$ the row (resp. columns) degrees. 
    \end{proof}
	Merge and swap updates for dc-LBM closely resemble those of dc-SBM and can be derived in the same fashion. Moreover, the prior parameter $\beta$ can be set using the same approach as for dc-SBM. However, dealing with bi-partitions induces some particular constraints for both the genetic and hierarchical algorithms. The next section details how they can be extended to co-clustering. 
	
	\subsection{Dealing with bipartitions}

\paragraph{Genetic algorithm} The hybrid algorithm presented in Section~\ref{sec:HybridAlgo} can be easily extended to the co-clustering problem. The latter simultaneously seeks for a partition of the $n$ rows and $p$ columns of a data matrix $\Obs \in \R^{n\times p}$. In this case, we work with a partition $\mathcal{P}$ of $\{1,\hdots, \nb +\dim\}$ with the additional constraints that it decomposes into two disjoint sets of clusters  that corresponds to a partition of $\{1,\hdots,\nb \}$ and $\{\nb +1,\hdots,\nb +\dim\}$ respectively (one for the rows and one for the columns): 
\begin{equation}
\mathcal{P} =\left\{\bC_1^r,\hdots,\bC_{\Kr}^r,\bC_1^c,\hdots,\bC_{\Kc}^c\right\} :\left\{\begin{array}{cl}
\bigcup_{k}\bC^r_k &= \{1,\hdots,\nb \},\\
\bigcup_{l}\bC^c_l &= \{\nb +1,\hdots,\nb +\dim\}
\end{array} \right. .
\end{equation}
\revision{This constraint can be easily incorporated} by defining $\ICLex(\mathcal{P})=-\infty$ for partitions that do not fulfill this constraint and by initializing the algorithm with admissible solutions. This is sufficient to ensure that the obtained solutions will also be compatible with the constraints, since the admissible set of partitions is closed under the crossover and mutation operations used by the algorithm. 

\paragraph{Hierarchical algorithm} Furthermore, the hierarchical methodology of \Cref{sec:Hierarchical} can also be easily extended to bi-partitions. Indeed, the LBM prior
\[
\p(\bPi \mid \alpha) = \Dir_{\Kr}(\bPirow \mid \alpha) \times \Dir_{\Kc}(\bPicol \mid \alpha),
\]
leaves a factorized integrated likelihood for $\p(\Clust \mid \alpha)$, with a common parameter $\alpha$. Thus, the $\ICLlin$ approximation of \Cref{eq:ICLlinclear} is still log-linear in $\alpha$ and writes: 
\begin{equation}
\ICLlin(\Clust, \alpha) = (\Kr -1) \log(\alpha) + (\Kc - 1) \log(\alpha) + I(\Clust),
\end{equation}
with $I(\Clust) = I(\Clustrow) + I(\Clustcol)$ the intercepts defined in Equation~\eqref{eq:intercept}. Hence, with the constraint that a merge cannot be done between rows and columns clusters, one can look for the best row or column fusion to do at each step, therefore building two dendrograms in parallel, with a shared ($\alpha_f)_f$ sequence.

	\printbibliography

@book{everitt2011cluster,
	title =     {Cluster Analysis, Fifth Edition (Wiley Series in Probability and Statistics)},
	author = {Everitt, Brian S. and Landau, Sabine and Leese, Morven},
	publisher = {Wiley},
	isbn =      {0470749911,9780470749913,9780470977804,9780470977811,9780470978443},
	year =      {2011},
	series =    {Wiley Series in Probability and Statistics},
	edition =   {5th}
}

@article{murtagh1984fitting,
	title={Fitting straight lines to point patterns},
	author={Murtagh, Fionn and Raftery, Adrian E},
	journal={Pattern recognition},
	volume={17},
	number={5},
	pages={479--483},
	year={1984},
	publisher={Elsevier}
}

@article{banfield1993model,
	title={Model-based Gaussian and non-Gaussian clustering},
	author={Banfield, Jeffrey D and Raftery, Adrian E},
	journal={Biometrics},
	pages={803--821},
	year={1993},
	publisher={JSTOR}
}

@article{fraley1998algorithms,
	title={Algorithms for model-based Gaussian hierarchical clustering},
	author={Fraley, Chris},
	journal={SIAM Journal on Scientific Computing},
	volume={20},
	number={1},
	pages={270--281},
	year={1998},
	publisher={SIAM}
}

@techreport{tessier2006evolutionary,
	title={Evolutionary latent class clustering of qualitative data},
	author={Tessier, Damien and Schoenauer, Marc and Biernacki, Christophe and Celeux, Gilles and Govaert, G{\'e}rard},
	year={2006}
}

@incollection{scrucca2016genetic,
	title={Genetic algorithms for subset selection in model-based clustering},
	author={Scrucca, Luca},
	booktitle={Unsupervised Learning Algorithms},
	pages={55--70},
	year={2016},
	publisher={Springer}
}

@article{andrews2013using,
	title={Using evolutionary algorithms for model-based clustering},
	author={Andrews, Jeffrey L and McNicholas, Paul D},
	journal={Pattern Recognition Letters},
	volume={34},
	number={9},
	pages={987--992},
	year={2013},
	publisher={Elsevier}
}

@article{sneath1957application,
	title={The application of computers to taxonomy},
	author={Sneath, Peter HA},
	journal={Microbiology},
	volume={17},
	number={1},
	pages={201--226},
	year={1957},
	publisher={Microbiology Society}
}

@article{sokal1958statistical,
	author = {Sokal, R. R. and Michener, C. D.},
	journal = {University of Kansas Science Bulletin},
	pages = {1409-1438},
	title = {A statistical method for evaluating systematic relationships},
	volume = {38},
	year = {1958}
}

@article{ward1963hierarchical,
	title={Hierarchical grouping to optimize an objective function},
	author={Ward Jr, Joe H},
	journal={Journal of the American statistical association},
	volume={58},
	number={301},
	pages={236--244},
	year={1963},
	publisher={Taylor \& Francis Group}
}

@book{mclachlan2007algorithm,
	title={The EM algorithm and extensions},
	author={McLachlan, Geoffrey J and Krishnan, Thriyambakam},
	volume={382},
	year={2007},
	publisher={John Wiley \& Sons}
}

@article{blei2017variational,
	title={Variational inference: A review for statisticians},
	author={Blei, David M and Kucukelbir, Alp and McAuliffe, Jon D},
	journal={Journal of the American Statistical Association},
	volume={112},
	number={518},
	pages={859--877},
	year={2017},
	publisher={Taylor \& Francis}
}

@book{gelman2004bayesian,
	title =     {Bayesian data analysis},
  author={Gelman, Andrew and Carlin, John B and Stern, Hal S and Dunson, David B and Vehtari, Aki and Rubin, Donald B},	publisher = {Chapman \& Hall/CRC},
	year =      {2004},
	edition =   {2nd ed}
}

@article{Come2015,
  title={Model selection and clustering in stochastic block models based on the exact integrated complete data likelihood},
  author={C\^{o}me, Etienne and Latouche, Pierre},
  journal={Statistical Modelling},
  volume={15},
  number={6},
  pages={564--589},
  year={2015},
  publisher={SAGE Publications Sage India: New Delhi, India},
}

@article{Newman2016,
  title = {Estimating the Number of Communities in a Network},
  author = {Newman, M. E. J. and Reinert, Gesine},
  journal = {Phys. Rev. Lett.},
  volume = {117},
  issue = {7},
  pages = {078301},
  numpages = {5},
  year = {2016},
  month = {Aug},
  publisher = {American Physical Society},
  doi = {10.1103/PhysRevLett.117.078301},
  url = {https://link.aps.org/doi/10.1103/PhysRevLett.117.078301}
}

@article{Riolo2017,
  title = {Efficient method for estimating the number of communities in a network},
  author = {Riolo, Maria A. and Cantwell, George T. and Reinert, Gesine and Newman, M. E. J.},
  journal = {Phys. Rev. E},
  volume = {96},
  issue = {3},
  pages = {032310},
  numpages = {12},
  year = {2017},
  month = {Sep},
  publisher = {American Physical Society},
  doi = {10.1103/PhysRevE.96.032310},
  url = {https://link.aps.org/doi/10.1103/PhysRevE.96.032310}
}

@article{schwarz1978estimating,
  title={Estimating the dimension of a model},
  author={Schwarz, Gideon},
  journal={The annals of statistics},
  volume={6},
  number={2},
  pages={461--464},
  year={1978},
  publisher={Institute of Mathematical Statistics}
}

@book{fruhwirth2019handbook,
  title={Handbook of mixture analysis},
  author={Fruhwirth-Schnatter, Sylvia and Celeux, Gilles and Robert, Christian P},
  year={2019},
  publisher={Chapman and Hall/CRC}
}

@article{peixoto2014hierarchical,
  title={Hierarchical block structures and high-resolution model selection in large networks},
  author={Peixoto, Tiago P},
  journal={Physical Review X},
  volume={4},
  number={1},
  pages={011047},
  year={2014},
  publisher={APS}
}

@book{bouveyron2019model,
  title={Model-Based Clustering and Classification for Data Science: With Applications in R},
  author={Bouveyron, Charles and Celeux, Gilles and Murphy, T Brendan and Raftery, Adrian E},
  volume={50},
  year={2019},
  publisher={Cambridge University Press}
}

@article{wang1987stochastic,
  title={Stochastic blockmodels for directed graphs},
  author={Wang, Yuchung J and Wong, George Y},
  journal={Journal of the American Statistical Association},
  volume={82},
  number={397},
  pages={8--19},
  year={1987},
  publisher={Taylor \& Francis Group}
}

@article{mariadassou2010uncovering,
  title={Uncovering latent structure in valued graphs: a variational approach},
  author={Mariadassou, Mahendra and Robin, St{\'e}phane and Vacher, Corinne},
  journal={The Annals of Applied Statistics},
  volume={4},
  number={2},
  pages={715--742},
  year={2010},
  publisher={Institute of Mathematical Statistics}
}

@article{akaike1974new,
  title={A new look at the statistical model identification},
  author={Akaike, Hirotugu},
  journal={IEEE transactions on automatic control},
  volume={19},
  number={6},
  pages={716--723},
  year={1974},
  publisher={Ieee}
}

@article{matias2014modeling,
  title={Modeling heterogeneity in random graphs through latent space models: a selective review},
  author={Matias, Catherine and Robin, St{\'e}phane},
  journal={ESAIM: Proceedings and Surveys},
  volume={47},
  pages={55--74},
  year={2014},
  publisher={EDP Sciences}
}

@article{zreik2016dynamic,
	TITLE = {{The dynamic random subgraph model for the clustering of evolving networks}},
	AUTHOR = {Zreik, Rawya and Latouche, Pierre and Bouveyron, Charles},
	URL = {https://hal.archives-ouvertes.fr/hal-01122393},
	JOURNAL = {{Computational Statistics}},
	PUBLISHER = {{Springer Verlag}},
	YEAR = {2016},
	DOI = {10.1007/s00180-016-0655-5},
	HAL_ID = {hal-01122393},
	HAL_VERSION = {v3},
}

@Article{Bertoletti2015,
author="Bertoletti, Marco
and Friel, Nial
and Rastelli, Riccardo",
title="Choosing the number of clusters in a finite mixture model using an exact integrated completed likelihood criterion",
journal="METRON",
year="2015",
month="Aug",
day="01",
volume="73",
number="2",
pages="177--199",
issn="2281-695X",
doi="10.1007/s40300-015-0064-5",
url="https://doi.org/10.1007/s40300-015-0064-5"
}

@article{Corneli2016,
title = "Exact ICL maximization in a non-stationary temporal extension of the stochastic block model for dynamic networks",
journal = "Neurocomputing",
volume = "192",
pages = "81 - 91",
year = "2016",
note = "Advances in artificial neural networks, machine learning and computational intelligence",
issn = "0925-2312",
doi = "https://doi.org/10.1016/j.neucom.2016.02.031",
url = "http://www.sciencedirect.com/science/article/pii/S0925231216002599",
author = "Marco Corneli and Pierre Latouche and Fabrice Rossi"
}

@article{Baudry2010,
  title={Combining Mixture Components for Clustering.},
  author={Jean-Patrick Baudry and Adrian E. Raftery and Gilles Celeux and Kenneth Lo and Rapha{\"e}l Gottardo},
  journal={Journal of computational and graphical statistics : a joint publication of American Statistical Association, Institute of Mathematical Statistics, Interface Foundation of North America},
  year={2010},
  volume={9 2},
  pages={332-353}
  }

@article{Dempster77,
 ISSN = {00359246},
 URL = {http://www.jstor.org/stable/2984875},
 author = {A. P. Dempster and N. M. Laird and D. B. Rubin},
 journal = {Journal of the Royal Statistical Society. Series B (Methodological)},
 number = {1},
 pages = {1--38},
 publisher = {[Royal Statistical Society, Wiley]},
 title = {Maximum Likelihood from Incomplete Data via the EM Algorithm},
 volume = {39},
 year = {1977}
}

@book{McLachlan2000,
title={Finite Mixture Models},
author={McLachlan, Geoffrey and Peel, David},
doi={10.1002/0471721182},
year={2000},
publisher={John Wiley \& Sons, Inc.},
serie={Wiley Series in Probability and Statistics}
}

@article{Biernacki2000,
author={Biernacki, Christophe and Celeux, Gilles and Govaert, Gerard},
title={Assessing a mixture model for clustering  with the integrated completed likelihood}, 
journal={IEEE Transaction on Pattern Analysis and Machine Intelligence},
volume={7},
year={2000},
pages={719-725}
}

@article{Biernacki2010,
author={Biernacki, Christophe and Celeux, Gilles and Govaert, Gerard},
title={Exact and monte carlo calculations of integrated likelihoods for the latent class model},
journal={Journal of Statistical Planning and Inference},
volume={140},
year={2010},
pages={2991-3002}
}

@article{Daudin2008,
author={Daudin, J. and Picard, F. and Robin, S.},
title={A mixture model for random graph},
journal={Statistics and computing},
volume={18},
year={2008},
pages={1-36}
}

@article{Wyse2017, 
title={Inferring structure in bipartite networks using the latent blockmodel and exact ICL}, 
volume={5}, 
DOI={10.1017/nws.2016.25}, 
number={1}, 
journal={Network Science}, 
publisher={Cambridge University Press}, 
author={Wyse, Jason and Friel, Nial and Latouche, Pierre}, 
year={2017}, 
pages={45–69}}

@InProceedings{Qin2013,
author={Qin, Tai and Rohe, Karl},
title={Regularized Spectral Clustering under the Degree-Corrected Stochastic Blockmodel},
year={2013},
booktitle={Proceedings of Nips},
}

@article{Bar2001,
    author = {Bar-Joseph, Ziv and Gifford, David K. and Jaakkola, Tommi S.},
    title = "{Fast optimal leaf ordering for hierarchical clustering }",
    journal = {Bioinformatics},
    volume = {17},
    number = {1},
    pages = {S22-S29},
    year = {2001},
    month = {06},
    issn = {1367-4803},
    url = {https://doi.org/10.1093/bioinformatics/17.suppl\_1.S22},
    eprint = {http://oup.prod.sis.lan/bioinformatics/article-pdf/17/suppl\_1/S22/726790/17S022.pdf},
}

@book{Eiben2003,
author={Eiben, A. E. and Smith, J. E.},
title={Introduction to Evolutionary Computing, 2$^{nd}$ Edition},
publisher={Springer-Verlag},
year={2004}
}

@misc{minka2000estimating,
  title={Estimating a Dirichlet distribution},
  author={Minka, Thomas},
  year={2000},
  publisher={Technical report, MIT}
}

@article{nowicki2001estimation,
  title={Estimation and prediction for stochastic blockstructures},
  author={Nowicki, Krzysztof and Snijders, Tom A B},
  journal={Journal of the American statistical association},
  volume={96},
  number={455},
  pages={1077--1087},
  year={2001},
  publisher={Taylor \& Francis}
}

@article{karrer2011stochastic,
  title={Stochastic blockmodels and community structure in networks},
  author={Karrer, Brian and Newman, Mark EJ},
  journal={Physical review E},
  volume={83},
  number={1},
  pages={016107},
  year={2011},
  publisher={APS}
}

@article{zhao2012consistency,
  title={Consistency of community detection in networks under degree-corrected stochastic block models},
  author={Zhao, Yunpeng and Levina, Elizaveta and Zhu, Ji},
  journal={The Annals of Statistics},
  volume={40},
  number={4},
  pages={2266--2292},
  year={2012},
  publisher={Institute of Mathematical Statistics}
}

@article {zhu2014,
	title = {{Oriented and degree-generated block models: generating and inferring communities with inhomogeneous degree distributions}},
	journal = {Journal of Complex Networks},
	volume = {2},
	number = {1},
	year = {2014},
	pages = {1{\textendash}18},
	doi = {10.1093/comnet/cnt011},
	url = {http://comnet.oxfordjournals.org/content/2/1/1.abstract},
	author = {Zhu, Yaojia and Yan, Xiaoran and Cristopher Moore}
}

@article{Newman2004,
  title = {Finding and evaluating community structure in networks},
  author = {Newman, M. E. J. and Girvan, M.},
  journal = {Phys. Rev. E},
  volume = {69},
  issue = {2},
  pages = {026113},
  numpages = {15},
  year = {2004},
  month = {Feb},
  publisher = {American Physical Society},
  doi = {10.1103/PhysRevE.69.026113},
  url = {https://link.aps.org/doi/10.1103/PhysRevE.69.026113}
}

@masterthesis{Cole1998,
author={Cole, R. M.},
title={Clustering with Genetic Algorithms},
institution={University of Western Australia},
country={Australia},
year={1998}
}

@ARTICLE{Hruschka2009, 
author={E. R. {Hruschka} and R. J. G. B. {Campello} and A. A. {Freitas} and A. C. {Ponce Leon F. de Carvalho}}, 
journal={IEEE Transactions on Systems, Man, and Cybernetics, Part C (Applications and Reviews)}, 
title={A Survey of Evolutionary Algorithms for Clustering}, 
year={2009}, 
volume={39}, 
number={2}, 
pages={133-155}, 
doi={10.1109/TSMCC.2008.2007252}, 
ISSN={}, 
month={March},}

@article{zhong2003unified,
  title={A unified framework for model-based clustering},
  author={Zhong, Shi and Ghosh, Joydeep},
  journal={Journal of machine learning research},
  volume={4},
  number={Nov},
  pages={1001--1037},
  year={2003}
}

@inproceedings{heller2005bayesian,
  title={Bayesian hierarchical clustering},
  author={Heller, Katherine A and Ghahramani, Zoubin},
  booktitle={Proceedings of the 22nd international conference on Machine learning},
  pages={297--304},
  year={2005},
  organization={ACM}
}

@article{govaert2010,
author = { Gérard   Govaert  and  Mohamed   Nadif },
title = {Latent Block Model for Contingency Table},
journal = {Communications in Statistics - Theory and Methods},
volume = {39},
number = {3},
pages = {416-425},
year  = {2010},
publisher = {Taylor & Francis},
doi = {10.1080/03610920903140197},

URL = { 
        https://doi.org/10.1080/03610920903140197
    
},
eprint = { 
        https://doi.org/10.1080/03610920903140197
    
}

}

@inproceedings{Adamic2005,
 author = {Adamic, Lada A. and Glance, Natalie},
 title = {The Political Blogosphere and the 2004 U.S. Election: Divided They Blog},
 booktitle = {Proceedings of the 3rd International Workshop on Link Discovery},
 series = {LinkKDD '05},
 year = {2005},
 isbn = {1-59593-215-1},
 location = {Chicago, Illinois},
 pages = {36--43},
 numpages = {8},
 url = {http://doi.acm.org/10.1145/1134271.1134277},
 doi = {10.1145/1134271.1134277},
 acmid = {1134277},
 publisher = {ACM},
 address = {New York, NY, USA},
}

@article{Gleiser2003,
author = {Gleiser, P. M. and Danon, L.},
title = {COMMUNITY STRUCTURE IN JAZZ},
journal = {Advances in Complex Systems},
volume = {06},
number = {04},
pages = {565-573},
year = {2003},
doi = {10.1142/S0219525903001067},

URL = { 
        https://doi.org/10.1142/S0219525903001067
    
},
eprint = { 
        https://doi.org/10.1142/S0219525903001067
    
}}

@Manual{Rcore,
    title = {R: A Language and Environment for Statistical Computing},
    author = {{R Core Team}},
    organization = {R Foundation for Statistical Computing},
    address = {Vienna, Austria},
    year = {2019},
    url = {https://www.R-project.org/},
}

@Article{Eddelbuettel2017,
  title = {{Extending 	extit{R} with 	extit{C++}: A Brief Introduction to 	extit{Rcpp}}},
  author = {Dirk Eddelbuettel and James Joseph Balamuta},
  journal = {PeerJ Preprints},
  year = {2017},
  month = {aug},
  volume = {5},
  pages = {e3188v1},
  issn = {2167-9843},
  url = {https://doi.org/10.7287/peerj.preprints.3188v1},
  doi = {10.7287/peerj.preprints.3188v1},
}

@Manual{Bates2019,
  title = {Matrix: Sparse and Dense Matrix Classes and Methods},
  author = {Douglas Bates and Martin Maechler},
  year = {2019},
  note = {R package version 1.2-17},
  url = {https://CRAN.R-project.org/package=Matrix},
}

@Article{Eddelbuettel2014,
  title = {RcppArmadillo: Accelerating R with high-performance C++ linear algebra},
  author = {Dirk Eddelbuettel and Conrad Sanderson},
  journal = {Computational	Statistics and Data Analysis},
  year = {2014},
  volume = {71},
  month = {March},
  pages = {1054--1063},
  url = {http://dx.doi.org/10.1016/j.csda.2013.02.005},
}

@Manual{Bengtsson2019,
  title = {future: Unified Parallel and Distributed Processing in R for Everyone},
  author = {Henrik Bengtsson},
  year = {2019},
  note = {R package version 1.13.0},
  url = {https://CRAN.R-project.org/package=future},
}

@article{vinh2010information,
  title={Information theoretic measures for clusterings comparison: Variants, properties, normalization and correction for chance},
  author={Vinh, Nguyen Xuan and Epps, Julien and Bailey, James},
  journal={Journal of Machine Learning Research},
  volume={11},
  number={Oct},
  pages={2837--2854},
  year={2010}
}

\end{document}